\pdfoutput=1
\documentclass[reprint,aps,onecolumn,superscriptaddress,showkeys,10ptd]{revtex4-2}

\usepackage[utf8]{inputenc}
\usepackage[english]{babel}
\usepackage[T1]{fontenc}
\usepackage{dsfont}
\usepackage{amsmath}
\usepackage{amssymb}
\usepackage{hyperref}
\usepackage{cleveref} 
\usepackage{bm} 
\usepackage{booktabs}
\usepackage{multirow}
\usepackage{amsthm}
\usepackage{algorithm2e}
\usepackage{graphicx}
\usepackage[dvipsnames]{xcolor}
\usepackage{svg}
\usepackage[normalem]{ulem}
\usepackage{tikz}
\usetikzlibrary{quantikz2}
\usepackage{lipsum}
\usetikzlibrary{fadings,shapes,arrows,shadows,positioning,backgrounds,fit}  

\newcommand{\cbb}[1]{\cos^{#1}({\beta})}
\newcommand{\sbb}[1]{\sin^{#1}({\beta})}
\newcommand{\cbf}[2]{\cos^{#1}\left(\frac{\beta}{{#2}}\right)}
\newcommand{\sbf}[2]{\sin^{#1}\left(\frac{\beta}{{#2}}\right)}

\newtheorem{lemma}{Lemma}

\Crefname{figure}{Fig.}{Figs.}
\Crefname{equation}{Eq.}{Eqs.}
\Crefname{section}{Section}{Sections}
\Crefname{appendix}{Appendix}{Appendices}

\DeclarePairedDelimiter\ceil{\lceil}{\rceil}
\DeclarePairedDelimiter\floor{\lfloor}{\rfloor}
\DeclareMathOperator{\mean}{mean}

\newsavebox{\boxA}
\savebox\boxA{{ \begin{quantikz}
\lstick{${A}$}  & \gate[2]{e^{-i\beta  \hat{S}^{-}_{AB}}} &\qw & \qw\\
 \lstick{${B}$} &  &\gate[2]{e^{-i\beta  \hat{S}^{-}_{BC}}} &\qw \\
\lstick{${C}$} & \qw & &\qw
\end{quantikz}

}}
\newsavebox{\boxB}
\savebox\boxB{\raisebox{-8ex}{
\begin{tikzpicture}
\begin{scope}[>={Stealth[black]},
              every node/.style={fill=white,circle},
              every edge/.style={draw=black,thick,fill}]
\node (A) at (0,0) {}; 
\path [->] (1,-2.2)edge(2,-2.2);
\path [->] (1,-2.2)edge(1.5,-1.2);
\path [->] (1,-2.2)edge(1.5,-3.2);
\end{scope}
\node at (1.8,-1.1) {$\boldsymbol{T}$};
\node at (2.5,-2.2) {$\boldsymbol{S}_{AB}$};
\node at (2,-3.2) {$\boldsymbol{S}_{BC}$};

\end{tikzpicture}

}}

\newsavebox{\boxC}
\savebox{\boxC}{ \begin{quantikz}
\lstick{${A}$}  &\qw & \gate[2]{e^{-i\beta  \hat{S}^{-}_{AB}}} &\qw & \qw\\
 \lstick{${B}$} & \gate[2]{e^{-i\beta  \hat{S}^{-}_{BC}}}&  &\gate[2]{e^{-i\beta  \hat{S}^{-}_{BC}}} &\qw \\
\lstick{${C}$} & & \qw & &\qw
\end{quantikz}
}
\newsavebox{\boxD}
\savebox\boxD{\raisebox{-8ex}{\begin{tikzpicture}
\begin{scope}[>={Stealth[black]},
              every node/.style={fill=white,circle},
              every edge/.style={draw=black,thick,fill}]
\node (A) at (0,0) {}; 
\path [->] (1,-2.2)edge(2,-1.8);
\path [->] (1,-2.2)edge(2,-2.6);
\path [->] (1,-2.2)edge(1.5,-1.2);
\path [->] (1,-2.2)edge(1.5,-3.2);
\end{scope}
\node at (1.8,-1.1) {$\boldsymbol{T}$};
\node at (2.5,-1.8) {$\boldsymbol{S}_{AB}$};
\node at (2.5,-2.6) {$\boldsymbol{S}_{BC}$};
\node at (2,-3.2) {$\boldsymbol{S}_{AC}$};

\end{tikzpicture}
}}

\newsavebox{\boxE}
\savebox\boxE{\raisebox{-8ex}{ \begin{quantikz}
\lstick{${j}$}   &\gate{R_z(2\gamma q ) } &\qw,
\end{quantikz}
}}
\newsavebox{\boxF}
\savebox\boxF{\raisebox{-8ex}{ \begin{quantikz}
\lstick{${j}$}   & \ctrl{1} &\qw &\ctrl{1} &\qw \\
\lstick{${k}$}  & \targ{} &\gate{R_z(2\gamma Q ) } &\targ{}& \qw
\end{quantikz}
}}

\newsavebox{\boxTa}
\savebox\boxTa{{ \begin{quantikz}
\lstick{${A}$}  &\qw & \gate[3]{e^{-i\beta  \hat{P}^{-}_{ABC}}} &\qw & \qw\\
 \lstick{${B}$} & \qw&  &\qw &\qw \\
\lstick{${C}$} & \gate[2]{e^{-i\beta  \hat{S}^{-}_{CD}}}& \qw & \gate[2]{e^{-i\beta  \hat{S}^{-}_{CD}}}&\qw \\
\lstick{${D}$} & & \qw & &\qw 
\end{quantikz}
}}
\newsavebox{\boxTb}
\savebox\boxTb{\raisebox{-8ex}{\begin{tikzpicture}
\begin{scope}[>={Stealth[black]},
              every node/.style={fill=white,circle},
              every edge/.style={draw=black,thick,fill}]
\node (A) at (0,0) {}; 
\path [->] (1,-2.2)edge(2,-1.8);
\path [->] (1,-2.2)edge(2,-2.6);
\path [->] (1,-2.2)edge(1.5,-1.2);
\path [->] (1,-2.2)edge(1.5,-3.2);
\end{scope}
\node at (1.8,-1.1) {$\boldsymbol{T}$};
\node at (2.5,-1.8) {$\boldsymbol{S}_{CD}$};
\node at (2.5,-2.6) {$\boldsymbol{P}_{ABC}$};
\node at (2,-3.2) {$\boldsymbol{P}_{ABD}$};

\end{tikzpicture}
}}

\newsavebox{\boxTc}
\savebox\boxTc{{ \begin{quantikz}
\lstick{${A}$}  &\qw & \gate[3]{e^{-i\beta  \hat{P}^{-}_{ABC}}} &\qw & \qw\\
 \lstick{${B}$} & \gate[3]{e^{-i\beta  \hat{S}^{-}_{BD}}}&  &\gate[3]{e^{-i\beta  \hat{S}^{-}_{BD}}} &\qw \\
\lstick{${C}$} & & \qw & &\qw \\
\lstick{${D}$} & & \qw & &\qw 
\end{quantikz}
}}
\newsavebox{\boxTd}
\savebox\boxTd{\raisebox{-8ex}{\begin{tikzpicture}
\begin{scope}[>={Stealth[black]},
              every node/.style={fill=white,circle},
              every edge/.style={draw=black,thick,fill}]
\node (A) at (0,0) {}; 
\path [->] (1,-2.2)edge(2,-1.8);
\path [->] (1,-2.2)edge(2,-2.6);
\path [->] (1,-2.2)edge(1.5,-1.2);
\path [->] (1,-2.2)edge(1.5,-3.2);
\end{scope}
\node at (1.8,-1.1) {$\boldsymbol{T}$};
\node at (2.5,-1.8) {$\boldsymbol{S}_{BD}$};
\node at (2.5,-2.6) {$\boldsymbol{P}_{ADC}$};
\node at (2,-3.2) {$\boldsymbol{P}_{ABD}$};

\end{tikzpicture}
}}

\newsavebox{\boxTe}
\savebox\boxTe{{\begin{quantikz}
\lstick{${A}$}  & \gate[4]{e^{-i\beta  \hat{S}^{-}_{AD}}} & \gate[3]{e^{-i\beta  \hat{P}^{-}_{ABC}}} &\gate[4]{e^{-i\beta  \hat{S}^{-}_{AD}}} & \qw\\
 \lstick{${B}$} &&  & &\qw \\
\lstick{${C}$} & & \qw & &\qw \\
\lstick{${D}$} & & \qw & &\qw 
\end{quantikz}
}}
\newsavebox{\boxTf}
\savebox\boxTf{\raisebox{-8ex}{\begin{tikzpicture}
\begin{scope}[>={Stealth[black]},
              every node/.style={fill=white,circle},
              every edge/.style={draw=black,thick,fill}]
\node (A) at (0,0) {}; 
\path [->] (1,-2.2)edge(2,-1.8);
\path [->] (1,-2.2)edge(2,-2.6);
\path [->] (1,-2.2)edge(1.5,-1.2);
\path [->] (1,-2.2)edge(1.5,-3.2);
\end{scope}
\node at (1.8,-1.1) {$\boldsymbol{T}$};
\node at (2.5,-1.8) {$\boldsymbol{S}_{AD}$};
\node at (2.5,-2.6) {$\boldsymbol{P}_{ABC}$};
\node at (2,-3.2) {$\boldsymbol{P}_{DBC}$};

\end{tikzpicture}
}}

\begin{document}

\title{Quantifying the advantages of applying quantum approximate algorithms to portfolio optimisation}

\author{\begingroup
\hypersetup{urlcolor=navyblue}
\href{https://orcid.org/0009-0001-4941-5448}{Haomu Yuan}
\endgroup}
\email[Haomu Yuan ]{hy374@cam.ac.uk}
\affiliation{Cavendish Laboratory, Department of Physics, University of Cambridge, Cambridge, CB3 0US, United Kingdom}

\author{\begingroup
\hypersetup{urlcolor=navyblue}
\href{https://orcid.org/0009-0001-3230-942X}{Christopher K. Long}
\endgroup}
\affiliation{Cavendish Laboratory, Department of Physics, University of Cambridge, Cambridge, CB3 0US, United Kingdom}

\author{\begingroup
\hypersetup{urlcolor=navyblue}
\href{https://orcid.org/0000-0002-7363-4165}{Hugo V. Lepage}
\endgroup}
\affiliation{Cavendish Laboratory, Department of Physics, University of Cambridge, Cambridge, CB3 0US, United Kingdom}

\author{\begingroup
\hypersetup{urlcolor=navyblue}
\href{https://orcid.org/0000-0001-7337-7245}{Crispin H. W. Barnes}
\endgroup}
\affiliation{Cavendish Laboratory, Department of Physics, University of Cambridge, Cambridge, CB3 0US, United Kingdom}

\date{\today}

\begin{abstract}
We present a quantum algorithm for portfolio optimisation. Specifically, We present san end-to-end quantum approximate optimisation algorithm to solve the discrete global minimum variance portfolio model. This model finds a portfolio of risky assets with the lowest possible risk contingent on the number of traded assets being discrete.  We  provide a complete pipeline for this model and analyse its viability for noisy intermediate-scale quantum computers. We design initial states, a cost operator, and ans\"atze within a binary encoding. Further, we perform numerical simulations to analyse several optimisation routines, including layerwise optimisation, utilising constrained optimisation by linear approximation and dual annealing. Finally, we consider the impacts of thermal relaxation and stochastic measurement noise. We find dual annealing with a layerwise optimisation routine provides the most robust performance. We observe that realistic thermal relaxation noise levels preclude quantum advantage. However, stochastic measurement noise will dominate when hardware sufficiently improves. Within this regime, we numerically demonstrate a favourable scaling in the number of shots required to obtain the global minimum---an indication of quantum advantage in portfolio optimisation.

\end{abstract}


\maketitle

\section{Introduction}
Variational quantum algorithms promise to be powerful tools for approximate optimisation of both classical and quantum problems~\cite{nielsenQuantumComputationQuantum2010, mcardleVariationalAnsatzbasedQuantum2019, cerezoVariationalQuantumAlgorithms2021}. The quantum approximate optimisation algorithm (QAOA) is one of the most well-known quantum approximate optimisers. It was first introduced to solve binary optimisation problems, MaxCut and MaxSAT, and was then applied to problems in chemistry, quantum field theory, and portfolio optimisation~\cite{farhiQuantumApproximateOptimization2014,hadfieldQuantumApproximateOptimization2019,eggerWarmstartingQuantumOptimization2021,bravyiHybridQuantumclassicalAlgorithms2022,farhiQuantumApproximateOptimization2022,Fernandez-Lorenzo_2021}.
Compared with other applications, quantum portfolio optimisation is arguably the most successful because it relates directly to real-world hedge funds and investment management.
Other quantum algorithms have been developed tackling portfolio optimisation with simple models~\cite{ peruzzoVariationalEigenvalueSolver2014, slateQuantumWalkbasedPortfolio2021,brandaoFasterQuantumClassical2022,rebentrostQuantumComputationalFinance2024}. However, many still require validation in the presence of quantum noise and stochastic noise using experiment or simulation. The QAOA has developed applications on more sophisticated portfolio optimisation models with a substantial body of simulations and benchmarks~\cite{hodsonPortfolioRebalancingExperiments2019,brandhoferBenchmarkingPerformancePortfolio2022,hermanSurveyQuantumComputing2022,chenQuantumPortfolioOptimization2023}. Experiments on current quantum devices suggest that QAOA is capable of reaching a regime where it could challenge the best classical algorithms~\cite{zhouQuantumApproximateOptimization2020}.

After Markowitz introduced the mean-variance portfolio model, portfolio optimisation became an important research field in financial management~\cite{markowitzPortfolioSelection1952, kolm60YearsPortfolio2014}. However, any departure of the asset return distribution from normality can significantly affect the mean-variance approximation to the optimal portfolio allocation~\cite{jondeauOptimalPortfolioAllocation2006}.
The complexity and interdependence of global economies make it challenging to predict economic stability, driving investors to pay more attention to portfolio risk.
The global minimum variance portfolio (GMVP) is a known risk-focus portfolio that yielded superior results in many empirical studies~\cite{haugenNewFinanceCase1994,clarkeMinimumvariancePortfoliosUS2006,clarkeMinimumvariancePortfolioComposition2011,chanPortfolioOptimizationForecasting2015}.
However, trading in a practical scenario is usually quantified by discrete numbers of assets because financial products such as stocks, options and futures are not usually infinitely divisible. The discrete GMVP (DGMVP) is the discretisation of the GMVP, in which the number of traded assets is discrete. The DGMVP is a special case of integer quadratic programming (IQP). As IQP is well known to be NP-hard~\cite{Chaovalitwongse2009}, a more efficient method to estimate the DGMVP is of great interest to investors and quantitative analysts.

A mathematical summarisation of the DGMVP model can be written as
\begin{equation}\label{qdp}
    \hspace*{-1.1cm} 
    \begin{aligned}
        \underset{\bm{w}}{\operatorname*{argmin}} \ & f(\bm{w}) = \bm{w}^\intercal \bm{\Sigma} \bm{w}
    \end{aligned}
\end{equation}
\vspace{-2.29em}
\begin{subequations}
\begin{align}
    \quad \mathrm{s.t.} \quad & \boldsymbol{w}^\intercal \mathbf{1}_n = 1, \tag{2a} \label{qdp:2a}\\
    \quad & {w}_i \in [0,1], \tag{2b} \label{qdp:2b}\\
    \quad & {w}_i/a \in \mathbb{Z}, \ i = 1,2,\dots,n \tag{2c} \label{qdp:2c}
\end{align}
\end{subequations}
where $\boldsymbol{w} = (w_1,w_2,\dots, w_n)^\intercal $ is an $n$-dimensional vector of weights for assets $1, \dots, n$, such as stocks, futures, and options, $\mathbf{1}_n$ is an $n$-dimensional ones vector, and $a$ is the unit trading lot for assets. The $n\times n$ covariance matrix $\boldsymbol{\Sigma}$ and its elements $\sigma_{ij}$ quantify the risk between each pair of assets. The first constraint \Cref{qdp:2a} is the budget constraint, the second constraint \Cref{qdp:2b} ensures no short investments, and the third constraint \Cref{qdp:2c} is the discretisation. 
In this article, we will present, to the best of our knowledge, the first end-to-end DGMVP solver based on the QAOA. 
We design a variational QAOA ansatz and quantify the use of quantum resources. The simulation analysis gives insight into solving other portfolio models and constrained mixed-IQP with quantum approximate algorithms in the future.

The rest of this article is structured as follows. In \Cref{sec:QAOAreview}, we review the original QAOA and general QAOA methods.
In \Cref{sec:QuantOptMeth}, we outline a set of QAOA ansatz designs and compare their performance in solving the DGMVP model.
In \Cref{sec:Simulation}, we compare different optimisers and optimisation methods to get the most efficient optimisation strategy.
In \Cref{sec:analysis}, we analyse the scaling of the quantum algorithm performance as a function of the asset number and the coarseness of the discretisation and compare it with classical scaling. 
In \Cref{sec:noise}, we add quantum decoherence noise and apply a post-selection method to improve the approximation.
Finally, we conclude in \Cref{sec:conclusion}.

\section{Review of QAOA}\label{sec:QAOAreview}
Farhi \textit{et al}. first introduced the quantum approximate optimisation algorithm (the original QAOA) to solve the combinatorial graph problem MaxCut and showed potential advantages~\cite{farhiQuantumApproximateOptimization2014}. 
The original QAOA was then generalised by the Quantum Alternating Operator Ansatz (the general QAOA) by redefining a family of cost and mixing operators~\cite{hadfieldQuantumApproximateOptimization2019}. To avoid ambiguity, we use the QAOA to represent the general QAOA in the rest of this article, which typically consists of:
\begin{itemize}
    \item [1.] \textbf{Initial state preparation.} We start with a pre-processing step to adjust the zero register to a desired state such as
\begin{equation}
    \begin{aligned}
         |\operatorname{Init}\rangle= {\hat{U}_I} |0\rangle^{\otimes n},\label{Eq:inistate}
    \end{aligned}
\end{equation}
where ${\hat{U}_I}$ is a pre-processing operator acting on the $|0\rangle ^{\otimes n}$ state. For example, in the MaxCut problem~\cite{farhiQuantumApproximateOptimization2014}, a Hadamard transformation is used to prepare the initial state:
\begin{align}
    {\hat{U}_I} = \hat{\operatorname{H}}^{\otimes n},  \quad     |\operatorname{Init}\rangle=\hat{\operatorname{H}}^{\otimes n}|0\rangle^{\otimes n}, \label{Eq:iniHstate} 
\end{align}
where $\hat{\operatorname{H}}$ is a Hadamard gate. A good initial state will increase the performance of quantum optimisation. Other initial states include the Hartree-Fock~\cite{doi:10.1126/science.abb9811,grimsleyAdaptiveProblemtailoredVariational2023}, Dicke~\cite{cookQuantumAlternatingOperator2020}, and warm-started states~\cite{eggerWarmstartingQuantumOptimization2021,Tate2023}. The warm-started preparation method indicates a strong advantage and beats the GW approximation in solving the MaxCut problem in a low-depth circuit~\cite{Tate2023}.
\item [2.] \textbf{Cost operator design.}
The cost operator $\hat{U}(\hat{C},\gamma)$ is defined such that 
\begin{align}
\hat{U}(\hat{C},\gamma) = e^{-i\gamma \hat{C}}, \label{Eq:MaxCutU1}
\end{align}
where $\gamma$ is the phase parameter and $\hat{C}$ is the cost Hamiltonian. The parameter $\gamma$ is usually restricted by a range. For example, $\gamma \in [0,2\pi]$ in the MaxCut problem. 
The cost Hamiltonian $\hat{C}$ encodes a cost function $f$, such as the one in \Cref{qdp}, as a series of Pauli-Z strings such that
\begin{align}
     \hat{C}|z\rangle = f(z) |z\rangle, \label{Eq:MaxCutH1}
\end{align}
where $|z\rangle$ is a computational basis state.
The cost operator is sometimes named a phase operator because it only generates a phase if acting on a state from the Pauli-Z eigenbasis.
\item [3.] \textbf{Mixing operator design.} 
A mixing operator is another parameterised operator alternating with the cost operator. In the MaxCut problem, the authors use a mixing operator $\hat{U}(\hat{B},\beta)$ with $\hat{B}$ defined by Pauli-X operators:
\begin{align}
     \hat{B} =\sum_{j=1}^n \hat{X}_j,  \quad \hat{U}(\hat{B},\beta) = e^{-i\beta \hat{B}} =\prod_{j=1}^{n}e^{-i \beta \hat{X}_j}, \label{Eq:MaxSatCutU2}
\end{align}
where $\beta\in[0,\pi]$, $n$ is the the number of qubits, and $\hat{X}_j$ is the Pauli-X operator acting on the $j$th qubit. The design of mixing operators is becoming a notable research field. A variety of mixing operators have been developed for a range of problems utilising sophisticated structures and providing other functionalities~\cite{hadfieldQuantumApproximateOptimization2019,wangMixersAnalyticalNumerical2020,bartschiGroverMixersQAOA2020}.
\item [4.] \textbf{Parameterised alternating ansatz construction.}
A parameterised alternating circuit consists of cost operators and mixing operators applied in an alternating fashion to the initial state as shown in \Cref{fig:ParaAlQAOA}. The output state will be
\begin{equation}\label{altCirc}
    \begin{aligned}
        |\bm{\gamma},\bm{\beta}\rangle\coloneqq \hat{U}(\hat{B},\beta_p) \hat{U}(\hat{C},\gamma_p)...\hat{U}(\hat{B},\beta_1) \hat{U}(\hat{C},\gamma_1)  |\operatorname{Init}\rangle,
    \end{aligned}
\end{equation}
where $p$ denotes the number of layers (one mixing operator and one cost operator) of the QAOA, $\bm{\gamma}= (\gamma_1, \gamma_2, ..., \gamma_p)$, and $\bm{\beta}=(\beta_1, \beta_2, ..., \beta_p)$. 
\item [5.] \textbf{Parameter optimisation.}
The parameter array $(\bm{\gamma}, \bm{\beta})$ is optimised by an objective function. We utilise the expectation value of the cost Hamiltonian as the objective function, but other common choices are Conditional Value-at-Risk, or variance~\cite{barkoutsosImprovingVariationalQuantum2020, zhangVariationalQuantumEigensolvers2022}.
If we consider a more realistic situation each measurement of the quantum circuit yields a computational basis state $|\bm{w}\rangle$ representing a classical value $\bm{w}$. The corresponding cost function value is $f(\bm{w})$, and each measurement constitutes a single function access.
After $k$ measurements, we can estimate the expectation value of cost Hamiltonian $\hat{C}$ as 
\begin{align}
\langle\hat{C} \rangle_{k} \coloneqq \langle\bm{\gamma},\bm{\beta}|\hat{C}|\bm{\gamma},\bm{\beta}\rangle_{k}=\frac{1}{k}{\sum_{i = 1}^{k} f(\bm{w}_i)}.  \label{appFp}
\end{align}
where such an estimation uses $k$ function accesses.
Classical optimisers are used to optimise the cost and mixing operators' parameters~\cite{PhysRevResearch.2.043246,Streif_2020,Fernandez-Lorenzo_2021,fernandez-pendasStudyPerformanceClassical2022,pellow-jarmanEffectClassicalOptimizers2024}. 
\end{itemize}

\section{Quantum optimisation methods} \label{sec:QuantOptMeth}
In this section, we first binary-encode the weight vector $\boldsymbol{w}$. Then, this section presents the three main elements of the QAOA within the encoding: initial states, a cost operator, and a {constrained} mixing operator.

\begin{figure}[h]
	\begin{center}
		\includegraphics{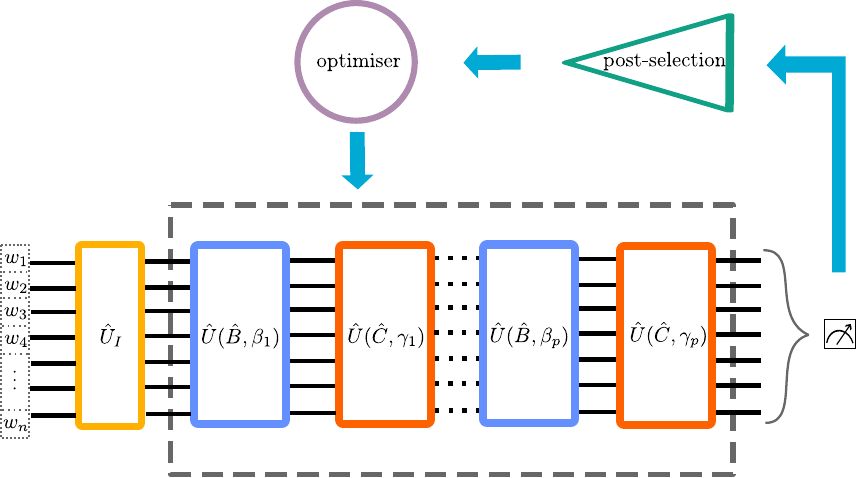}
		\caption{A parameterised QAOA ansatz for predicting the {optimal} DGMVP solution. The variables $w_1,w_2,\dots, w_n$ in the dotted squares are the binary-encoded weight variables. The unitary ${\hat{U}_I}$ prepares the initial state. The grey dashed box includes the mixing [$\hat{U}(\hat{B},\beta]$)] and cost [$\hat{U}(\hat{C},\gamma)$)] operators. A tunable parameter controls each operator. After measurements, a post-selection process will be applied to filter out unconstrained results and feed them to the classical optimiser.}\label{fig:ParaAlQAOA}
	\end{center}
\end{figure} 

\subsection{Binary block encoding method}\label{sec:binaryblockenco}
In previous quantum portfolio optimisation algorithms~\cite{eggerWarmstartingQuantumOptimization2021, hodsonPortfolioRebalancingExperiments2019, slateQuantumWalkbasedPortfolio2021}, the authors encode at most three trading positions (long, short, and on-hold) with two binary variables by treating the first variable as a long position and the second as a short position. 
However, the DGMVP and most portfolio strategies usually require an encoding with more positions encoded as multiples of the unit trading lot $a$.
A quasi-binary encoding method provided in~\cite{chenQuantumPortfolioOptimization2023} uses a large enough factor $1/a$ to amplify all weights variables to integers. 
This article uses a binary encoding method with a fixed length for all assets. 

\begin{figure}[h]
	\begin{center}
		\includegraphics{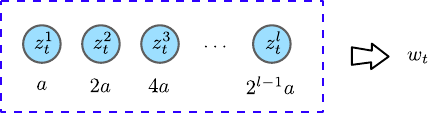}
		\caption{A diagram of binary block encoding for the $t$th asset investment weight $w_t$. The circles represent binary variables $z_t^1, z_t^2,\dots, z_t^l$ from measuring qubits on the Pauli-Z basis. The weight of variables $a, 2a, \dots, 2^{l-1} a$ are shown at the bottom of the corresponding variables circles. Each $w_t$ combines an $l$ length block of binary variables from the measurement of qubits. }\label{fig:binaryenc}
	\end{center}
\end{figure}

The variable $w_t$ can be encoded with a block of $l$ binary variables (see \Cref{fig:binaryenc}),
\begin{equation}\label{eq:measureNot}
	\begin{aligned}
		w_t = \sum_{k=1}^{l} b^k z^k_{t},
	\end{aligned}
\end{equation}
where $z^k_{t}\in{ \{0,1\}}$ is a binary variable corresponding to the measurement of a qubit state in Pauli-Z basis, and $b^k$ is 
\begin{equation}\label{coeffic}
	\begin{aligned}
		b^{k} = 2^{k-1} a,
	\end{aligned}
\end{equation}
where $a$ is the binary precision and can be interpreted as the unit trading lot:
\begin{equation}\label{accuracy}
	a =\frac{1}{ \sum_{i=0}^{l-1} 2^i}=\frac{1}{2^l-1}.
\end{equation}
Then, the total number of shares is $1/a$. 
As block sizes increase, the binary encoding method will use exponentially fewer qubits than the one-hot encoding method~\cite{chancellorDomainWallEncoding2019}. 

\subsection{Initial state preparations}\label{sec:initialstate} 
A good design of an initial state for target applications can avoid initialising far away from the global optimum and reduce the convergence time in quantum optimisation~\cite{eggerWarmstartingQuantumOptimization2021, Tate2023}.
This article provides four initial state preparation methods for the DGMVP model: max-biased state, ranked warm-started state, approximated equal-weighted state, and random-weighted state. Simulations in \Cref{app:inistat} show max-biased state and ranked warm-started state outperform other methods. Thus, we focus on these two methods in the body and summarise the rest of the initial states in the \Cref{app:inistat}. 

\textit{The max-biased state $|\boldsymbol{w}\rangle_{mb}$:---}This state represents an all-in-one investment without diversification, as  
\begin{equation}\label{eq:maxbias}
    w_i = \left\{\begin{array}{l}
		1, i = t'\\
		0, i\neq t'.
	   \end{array}\right.
\end{equation}

\textit{The ranked warm-started state $|\boldsymbol{w}\rangle_{ws}$:---}This state represents an approximation to the {optimal} DGMVP solution from the round-off of the GMVP model's 
 continuous solution $\tilde{\boldsymbol{w}}$. The GMVP can be solved using a semidefinite programming solver. 
The values in $\tilde{\boldsymbol{w}}$ can be divided by the binary precision $a$ to get sets of reminders $\tilde{\boldsymbol{r}}$ and divisors $\tilde{\boldsymbol{z}}$: 
\begin{equation}\label{eq:cw}
	\begin{aligned}
		\tilde{z}_i = \floor*{\frac{\tilde{w}_i}{a}}, \quad \tilde{r}_i = \frac{\tilde{w}_i}{a} - \floor*{\frac{\tilde{w}_i}{a}}, \quad \forall \tilde{w}_i\in \tilde{\boldsymbol{w}}.
	\end{aligned}
\end{equation}
Further, we define the remaining budget $u\coloneqq 1/a - \sum_{\tilde{z}_i\in \tilde{\boldsymbol{z}}} \tilde{z}_i$.
First, we order the assets from the largest value $\tilde{r}_i$ to the smallest.
Then, we assign the remaining budget one-by-one to the first $u$ assets. 
A ranked warm-started state $|\boldsymbol{w}\rangle_{ws}$ of the DGMVP model prepares the asset weights as 
\begin{equation}\label{eq:ws}
	w_{i} = \left\{\begin{array}{l}
		a\tilde{z}_{i} +a, i\leq u\\
		a\tilde{z}_{i}, i>u
		\end{array}\right. .
\end{equation}
Finding continuous solution $\tilde{\boldsymbol{w}}$ corresponds to solving a convex quadratic program which can be achieved in polynomial time in $n$~\cite{boyd2004convex}. The sorting can be achieved in $\mathcal{O}(n\log n)$ by the quick sort algorithm~\cite{cormen2022introduction}.

\subsection{The design of cost Hamiltonian and cost operator for the DGMVP model}
This section introduces the construction of the cost Hamiltonian $\hat{C}$ for the DGMVP model, the corresponding cost operator, and the quantum circuit.
Using the binary block encoding from \Cref{sec:binaryblockenco}, the DGMVP problem in \Cref{qdp} can be rewritten as
\begin{equation}\label{transf}
    \begin{aligned} 
        \begin{array}{ll}
            \underset{\bm{x}}{\operatorname*{argmin}}  & f  = a^2 \sum_{i, j} \sigma_{ij}  x_i x_j  \\
             \quad \mathrm{s.t.} & \sum_{i} a x_i=1, \\
             & 0 \leq x_i \leq D, \ x_i \in \mathbb{Z}
        \end{array}
    \end{aligned}
\end{equation}
where $\bm{x}$ is a $n$ dimensional vector with elements $x_i = w_i/a = \sum_{k=1}^l b^k z_i^k$ (see in \Cref{sec:binaryblockenco}) and $D= 1/a$. Substituting for the $w_i$, we find
\begin{equation}\label{transfz}
    \begin{aligned}
          & f= 2 a^2 \sum_{i<j}\sum_{k_1,k_2}\sigma_{ij}b^{k_1}z_i^{k_1} b^{k_2} z_j^{k_2} + a^2 \sum_{i}\sum_{k_1,k_2}\sigma_{ii}b^{k_1}z_i^{k_1} b^{k_2} z_i^{k_2}.  
    \end{aligned}
\end{equation}

Using the relations
\begin{equation}\label{pauliBinary} 
		\frac{\hat{\mathds{1}}-\hat{Z}}{2} |0\rangle = 0 |0\rangle, \frac{\hat{\mathds{1}}-\hat{Z}}{2} |1\rangle = 1 |1\rangle,
\end{equation}
where $\hat{Z}$ is a Pauli-Z operator and $\hat{\mathds{1}}$ is the identity operator, and~\Cref{Eq:MaxCutU1}, the Hamiltonian $\hat{C}$ can be written as 
\begin{equation}\label{addz}
	\begin{aligned}
  		\hat{C} &=\frac{a^2}{2}\sum_{i<j}\sum_{k_1,k_2} \sigma_{ij} b^{k_1} b^{k_2} \left(\hat{Z}^{k_1}_{i} \hat{Z}^{k_2}_{j} -\hat{Z}^{k_1}_{i} - \hat{Z}^{k_2}_{j} + \hat{\mathds{1}}\right)\\  	
  			& \quad + \frac{a^2}{4} \sum_{i}\sum_{k_1,k_2}  \sigma_{ii}b^{k_1} b^{k_2}  \left(\hat{Z}^{k_1}_{i} \hat{Z}^{k_2}_{i} - \hat{Z}^{k_1}_{i} - \hat{Z}^{k_2}_{i} +\hat{\mathds{1}}\right),
	\end{aligned}
\end{equation}
where $\hat{Z}^{k}_i$ is the Pauli-Z operator on the $k$th qubit of the $i$th asset.  
Since every weight operator use the same number of binary variables for the encoding, then 
\begin{equation}\label{eq:costrecombsupp3}
	\begin{aligned}
  		\frac{a^2}{4}\sum_{i}\sum_{k_1,k_2} \sigma_{ii}b^{k_1} b^{k_2}\hat{Z}^{k_2}_{i} = \frac{a^2}{4}\sum_{i}\sum_{k_1,k_2} \sigma_{ii}b^{k_1} b^{k_2}\hat{Z}^{k_1}_{i}
	\end{aligned}
\end{equation}
and
\begin{equation}\label{eq:costrecombsupp1}
	\begin{aligned}
  		& \quad \frac{a^2}{2}\sum_{i<j}\sum_{k_1,k_2} \sigma_{ij}b^{k_1} b^{k_2}\hat{Z}^{k_1}_{i}  +\frac{a^2}{2}\sum_{i<j}\sum_{k_1,k_2} \sigma_{ij}b^{k_1} b^{k_2}\hat{Z}^{k_2}_{j}  + \frac{a^2}{2}\sum_{i}\sum_{k_1,k_2} \sigma_{ii}b^{k_1} b^{k_2}\hat{Z}^{k_2}_{i} \\
		& = \frac{a^2}{2}\sum_{i<j}\sum_{k_1,k_2} \sigma_{ij}b^{k_1} b^{k_2}\hat{Z}^{k_1}_{i}  +\frac{a^2}{2}\sum_{j<i}\sum_{k_1,k_2} \sigma_{ij}b^{k_1} b^{k_2}\hat{Z}^{k_1}_{i}  + \frac{a^2}{2}\sum_{i}\sum_{k_1,k_2} \sigma_{ii}b^{k_1} b^{k_2}\hat{Z}^{k_1}_{i} \\
		& = \frac{a^2}{2} \sum_{i,j} \sum_{k_1,k_2} \sigma_{ij}b^{k_1} b^{k_2}\hat{Z}^{k_1}_{i} = \frac{a^2}{2} \sum_{i,j} \sum_{k_2}b^{k_2}\sum_{k_1} \sigma_{ij} b^{k_1} \hat{Z}^{k_1}_{i} = \frac{a^2}{2} \sum_{i,j} \frac{1}{a}\sum_{k_1} \sigma_{ij} b^{k_1} \hat{Z}^{k_1}_{i} \\
		& = \frac{a}{2} \sum_{i} \sum_{k} \tilde{\sigma}_{i} b^{k} \hat{Z}^{k}_{i},
	\end{aligned}
\end{equation}
where the second equality comes from the symmetry $\sigma_{ij} = \sigma_{ji}$, and $\tilde{\sigma}_{i}= \sum_j\sigma_{ij}$. By substituting the \Cref{eq:costrecombsupp3,eq:costrecombsupp1} into \Cref{addz}, the DGMVP Hamiltonian $\hat{C}$ can be expressed as 
\begin{equation}\label{eq:hamiltonian}
	\begin{aligned}
  		\hat{C} &=\frac{a^2}{2}\sum_{i<j}\sum_{k_1,k_2} \sigma_{ij}b^{k_1} b^{k_2} \hat{Z}^{k_1}_{i} \hat{Z}^{k_2}_{j} + \frac{a^2}{4} \sum_{i}\sum_{k_1\neq k_2}\sigma_{ii}b^{k_1} b^{k_2} \hat{Z}^{k_1}_{i} \hat{Z}^{k_2}_{i} - \frac{a}{2} \sum_{i} \sum_{k}\tilde{\sigma}_{i} b^{k} \hat{Z}^{k}_{i}  +c,\\
	\end{aligned}
\end{equation}
where $c$ is a constant scalar, the $k_1\neq k_2$ condition in the second term is derived from $\hat{Z}^{k_1}_{i} \hat{Z}^{k_2}_{i} = \hat{\mathds{1}}$ when $k_1 = k_2$. We can drop the constant offset $c$ because {it only yields a global phase}. The cost operator of the DGMVP Hamiltonian $\hat{C}$ is
\begin{equation}\label{eq:Cz1}
	\begin{aligned}
		\hat{U}(\hat{C}, \gamma)  &= e^{-i \gamma \hat{C}}\\
		& =  \prod_{i<j}\prod_{k_1,k_2} e^{-\frac{i \gamma a^2}{2} \sigma_{ij}b^{k_1} b^{k_2} \hat{Z}^{k_1}_{i} \hat{Z}^{k_2}_{j}}\prod_{i}\prod_{k_1\neq k_2} e^{-\frac{i \gamma a^2}{4} \sigma_{ii}b^{k_1} b^{k_2} \hat{Z}^{k_1}_{i} \hat{Z}^{k_2}_{i}} \prod_{i}\prod_{k} e^{\frac{i \gamma a}{2} \tilde{\sigma}_{i} b^{k} \hat{Z}^{k}_{i}},
	\end{aligned}
\end{equation}
where $\gamma\in[0,2\pi]$. We can separate the products of exponentials because Pauli-Z strings commute with each other. \Cref{eq:Cz1} contains both singular and coupling Pauli-Z strings---the corresponding digital quantum circuit is shown in \Cref{costcirc}(a) and (b), respectively.
As the covariance matrix $\Sigma$ is dense, coupling Pauli-Z strings in {$\hat{U}(\hat{C}, \gamma)$} will contain interactions between all pairs of qubits. Thus, a dense quantum circuit can be constructed with $\mathcal{O}({N})$ depth, where $N$ is the number of qubits~\cite{PhysRevResearch.6.013254, PhysRevA.109.042413}.

\begin{figure}[h]
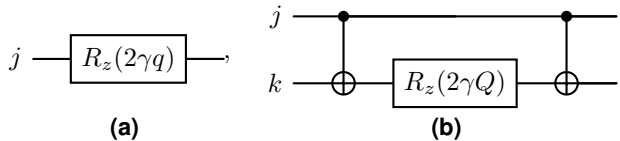

	\centering
	\begin{tabular}{c c}
		\usebox\boxE & \usebox\boxF \\
		\textbf{\fontfamily{phv}\selectfont (a)} & \textbf{\fontfamily{phv}\selectfont (b)}
	\end{tabular}
	\caption[Cost operators circuits]{(a) Cost operator circuit design of {$e^{-i\gamma q \hat{Z}_j}$}. (b) Cost operator circuit design of {$e^{-i\gamma Q \hat{Z}_j\hat{Z}_k}$}. } \label{costcirc}
\end{figure}

\subsection{{Constrained} mixing operator based on qubit excitations}\label{sec:mix}
In this section, we design a {constrained} mixing operator based on the qubit excitation operators to conserve the budget constraint in the QAOA optimisation. 
{We define} two types of {constrained} mixing operators in the current quantum portfolio solvers: \textit{Hard{-constraint} mixing operators} and \textit{soft{-constraint} mixing operators}, {for problems with hard and soft constraints, respectively \cite{hadfieldQuantumApproximateOptimization2019}.}
{{Hard-constraint mixing operators, unlike soft-constraint mixing operators, obey the problem constraints.}} 
In \Cref{appendix:scaling}, we show that the gap between the constrained and unconstrained feasible region size scales exponentially. This gap may imply that soft{-constraint} mixing operators can induce significant measurement overheads to optimisation. {Hadfield \textit{et al}. designed the \textit{parity mixer}, a hard{-constraint} mixing operator for preserving Hamming weights within the one-hot encoding~\cite{hadfieldQuantumApproximateOptimization2019}. Hodson \textit{et al}. alternated the application of the parity mixer between the odd and even spin subsets---representing short and long positions, respectively~\cite{hodsonPortfolioRebalancingExperiments2019}. They showed that the parity mixer outperforms the standard QAOA soft{-constraint} mixing operator in portfolio rebalancing.
However, the one-hot encoding will require an exponential number of qubits for the same precision as the binary block encoding \cite{chancellorDomainWallEncoding2019}---see~\Cref{sec:binaryblockenco} for details.
Chen \textit{et al}. extended the design of the parity mixer as the ring mixer by adding XYY mixers to solve a quasi-binary encoded problem~\cite{chenQuantumPortfolioOptimization2023}.}
Our algorithm includes a hard{-constraint} mixing operator design constructed using two- and three-qubit excitation operators and the induced quantum bridge phenomenon introduced in \Cref{app:three-qubitex,app:two-qubitex}.

In principle, within the binary encoding of the DGMVP model, a combination of two binary arithmetic operations can reach all constrained $\boldsymbol{w}$.
The first operation, \textit{exchange}, swaps excited qubits representing the variables with the same binary coefficient, such as $z_t^k$ and $z_{t'}^k$, where $z_{t}^k$ and $z_{t'}^{k}$ are the $k$th qubits in the $t$ and $t'$ asset blocks of the binary block encoding.
This conserves the total number of excited qubits. The exchange operation includes two transformations:
\newcommand{\ztk}{\raisebox{-0.2ex}{${}_{z_{t}^k}$}}
\newcommand{\ztm}{\raisebox{-0.2ex}{${}_{z_{t'}^k}$}}
\newcommand{\ztkt}{\raisebox{-0.2ex}{$_{z_{t}^{k+1}}$}}
\newcommand{\ztmt}{\raisebox{-0.2ex}{$_{z_{t''}^{k}}$}}
\begin{equation}\label{mixereq2}
	\renewcommand{\arraystretch}{1.5}
	\left\{\begin{array}{l}
		\ (\underset{\ztk}{1}, \, \underset{\ztm}{0} )\mapsto(\underset{\ztk}{0}, \, \underset{\ztm}{1})\\
		\ (\underset{\ztk}{0}, \, \underset{\ztm}{1})\mapsto(\underset{\ztk}{1}, \, \underset{\ztm}{0})
		\end{array}\right..
\end{equation}

The other operation is the \textit{carry-borrow} operation and includes transformations as follows:
\begin{equation}\label{mixereq1}
	\renewcommand{\arraystretch}{1.5}
	\left\{\begin{array}{l}
		\ (\underset{\ztkt}{0}\!\!,\, \underset{\ztm}{1},\underset{\ztmt}{1}) \mapsto (\underset{\ztkt}{1}\!\!, \, \underset{\ztm}{0}, \underset{\ztmt}{0}) \quad \textit{carry} \\
		\ (\underset{\ztkt}{1}\!\!, \, \underset{\ztm}{0}, \underset{\ztmt}{0}) \mapsto (\underset{\ztkt}{0}\!\!, \, \underset{\ztm}{1},\underset{\ztmt}{1}) \quad \textit{borrow}
		\end{array}\right.,
\end{equation}
where assets label $t'\neq t''$. The carry in the quantum setup corresponds to exciting a qubit at binary position $k+1$ and de-exciting two other qubits at binary position $k$, and the borrowing process is the reverse direction. 
\renewcommand{\arraystretch}{1}

{We implement these transformations with Givens rotations. The corresponding rotation for the exchange operation \Cref{mixereq2} is generated by a two-qubit excitation operator~\cite{yordanovEfficientQuantumCircuits2020}:}
{\begin{equation}\label{mixerOp1}
	\begin{aligned}
		 & \hat{S}_{tt'}^k \coloneqq i \left(\hat{Q}^{k\dagger }_{t}\hat{Q}^k_{t'} - \hat{Q}^k_{t}\hat{Q}^{k\dagger}_{t'}\right),
	\end{aligned}
\end{equation}}
where 
\begin{equation}\label{qubitcreaanni}
	\begin{aligned}
		\hat{Q}^{k\dagger }_{t} = \frac{1}{2}(\hat{X}^k_t-i \hat{Y}^k_t), \quad \hat{Q}^k_t = \frac{1}{2}(\hat{X}^k_t+i \hat{Y}^k_t),
	\end{aligned}
\end{equation}
and $\hat{X}^k_t$ and $\hat{Y}^k_t$ are the Pauli-X and Pauli-Y operators acting on $k$th qubit of the $t$th asset.
{Then, \Cref{mixerOp1} can be written as the Pauli operators:
\begin{equation}\label{eq:mixerOp1_pauli}
	\begin{aligned}
		 & \hat{S}_{tt'}^k =-\frac{1}{2}(\hat{X}^{k}_{t}\hat{Y}^{k}_{t'}- \hat{Y}^{k}_{t}\hat{X}^{k}_{t'}).
	\end{aligned}
\end{equation}}

{Similarly, the Givens rotation for the carry-barrow operation in~\Cref{mixereq1} is generated by the three-qubit excitation operator}
{
\begin{equation}\label{mixerOp2}
	\begin{aligned}
		  \hat{P}_{tt't''}^{k} &\coloneqq i \left(\hat{Q}^{{k+1}\dagger}_{t} \hat{Q}^{k}_{t'} \hat{Q}^{k}_{t''}- \hat{Q}^{{k+1}}_{t} \hat{Q}^{k\dagger}_{t'} \hat{Q}^{k\dagger}_{t''}\right) \\
         &= -\frac{1}{4} (\hat{X}^{k+1}_{t}\hat{X}^{k}_{t'}\hat{Y}^{k}_{t''}+\hat{X}^{k+1}_{t}\hat{Y}^{k}_{t'}\hat{X}^{k}_{t''}- \hat{Y}^{k+1}_{t}\hat{X}^{k}_{t'}\hat{X}^{k}_{t''}+\hat{Y}^{k+1}_{t}\hat{Y}^{k}_{t'}\hat{Y}^{k}_{t''}),
	\end{aligned}
\end{equation}
}
{where assets label $t'\neq t''$.}

{Using these qubit excitation operators we can express the Givens rotations as:}
{
\begin{equation}\label{mixerOpExps}
	\begin{aligned}
		& e^{-i \beta  \hat{S}^{k}_{tt'}} = \cos^2\frac{\beta}{2}\hat{\mathds{1}} + \sin^2\frac{\beta}{2}\hat{Z}^k_t \hat{Z}^k_{t'} -i  \sin\beta  \hat{S}^{k}_{tt'} \\  
            & e^{-i \beta  \hat{P}^{k}_{tt't''}}= \frac{1}{4}\left(3+\cbb{}\right)\hat{\mathds{1}} +\frac{1}{2}\sbf{2}{2}\left(\hat{Z}^{k+1}_{t}\hat{Z}^{k}_{t'} +\hat{Z}^{k+1}_{t}\hat{Z}^{k}_{t''}- \hat{Z}^{k}_{t'}\hat{Z}^{k}_{t''}\right) -i\sbb{}  \hat{P}^{k}_{tt't''},
	\end{aligned}
\end{equation}}
where {$\beta \in [0,\pi]$} and the corresponding quantum circuit design is shown in~\Cref{fig:digiEABEAAB}.           
Thus, {$e^{-i\beta  \hat{S}^{k}_{tt'}}$} has support on {$\hat{\mathds{1}}$, $\hat{Z}^k_t \hat{Z}^k_{t'}$}, and the two-qubit excitation operators {$\hat{S}^{k}_{tt'}$}. Additionally, {$e^{-i\beta  \hat{P}^{k}_{tt't''}}$} has support on {$\hat{\mathds{1}}$, $\hat{Z}^{k+1}_{t}\hat{Z}^{k}_{t'}, \ \hat{Z}^{k+1}_{t}\hat{Z}^{k}_{t''}$, $\hat{Z}^{k}_{t'}\hat{Z}^{k}_{t''}$}, and the three-qubit excitation operator {$\hat{P}^{k}_{tt't''}$}. 
The $\hat{\mathds{1}}$ leaves the state invariant. The Pauli-Z strings generate phases between computational basis states. The qubit excitation operators, {$\hat{S}^{k}_{tt'}$ and $\hat{P}^{k}_{tt't''}$}, perform the binary arithmetic operations in~\Cref{mixereq2,mixereq1}{, respectively}. The support on the qubit excitation operators changes the occupation of computational basis states while still obeying the constraints in~\Cref{qdp:2a,qdp:2b,qdp:2c}. The parameter $\beta$ controls the degree of superposition, performing optimisation improves the superposition to a better approximation.
In \Cref{fig:SPop} (a) and (b), we present two graphs that represent the support of the two- and three-qubit excitations, respectively.

{Additionally, we could add a phase rotation to the Givens rotations and still satisfy the constrains. For example, the parity mixer \cite{hadfieldQuantumApproximateOptimization2019}
\begin{equation}\label{eq:parity_mixer}
	\begin{aligned}
		  \hat{T}^k_{tt'} & \coloneqq \frac{1}{2}(\hat{X}^{k}_{t}\hat{X}^{k}_{t'}+ \hat{Y}^{k}_{t}\hat{Y}^{k}_{t'}) \\
         & = \hat{Q}^{k\dagger }_{t}\hat{Q}^k_{t'} + \hat{Q}^k_{t}\hat{Q}^{k\dagger}_{t'},
	\end{aligned}
\end{equation}
and XYY mixer~\cite{chenQuantumPortfolioOptimization2023}
\begin{equation}\label{eq:paulixyy}
	\begin{aligned}
		  \hat{O}_{tt't''}^{k} &\coloneqq \frac{1}{4}( \hat{X}^{k+1}_{t} \hat{X}^{k}_{t'} \hat{X}^{k}_{t''}+ \hat{Y}^{k+1}_{t} \hat{X}^{k}_{t'} \hat{Y}^{k}_{t''}- \hat{X}^{k+1}_{t} \hat{Y}^{k}_{t'} \hat{Y}^{k}_{t''}+\hat{Y}^{k+1}_{t} \hat{Y}^{k}_{t'} \hat{X}^{k}_{t''})  \\
         &=\hat{Q}^{{k+1}\dagger}_{t} \hat{Q}^{k}_{t'} \hat{Q}^{k}_{t''}+ \hat{Q}^{{k+1}}_{t} \hat{Q}^{k\dagger}_{t'} \hat{Q}^{k\dagger}_{t''}.
	\end{aligned}
\end{equation}
However, for simplicity we employ the Givens rotations and only generate phase in our ansatz with the cost operator.} 

Furthermore, if we compose the two- and three-qubit excitation operators as follows:
{\begin{equation}\label{eq:robot}
	\begin{aligned}
		 & e^{-i\beta  \hat{S}^{k'}_{tt'}} e^{-i\beta  \hat{P}^{k'}_{ttt'} } e^{-i\beta  \hat{S}^{k'}_{tt'}},
	\end{aligned}
\end{equation}}
it will also include the support of qubit excitation operators of {$e^{-i\beta  \hat{P}^{k'}_{tt't} }$} as shown from \Cref{fig:SPop} (c) to (d) due to a phenomenon we denote as the \textit{quantum bridge} outlined in \Cref{app:three-qubitex}. 
Thus, by stacking qubit excitations, the composition
{\begin{equation}\label{eq:robot2}
	\begin{aligned}
		& \left(\bigotimes_{i =k',k''} e^{-i\beta  \hat{S}^{i}_{tt'}}\right) \Biggl(e^{-i\beta  \hat{P}^{k'}_{ttt'} }e^{-i\beta  P^{k''}_{ttt'} } \Biggl)\left(\bigotimes_{i =k',k''} e^{-i\beta  \hat{S}^{i}_{tt'}}\right)        \\
		& =  \left( e^{-i\beta  \hat{S}^{k'}_{tt'}} e^{-i\beta  \hat{P}^{k'}_{ttt'} } e^{-i\beta  \hat{S}^{k'}_{tt'}} \right) \left( e^{-i\beta  \hat{S}^{k''}_{tt'}} e^{-i\beta  \hat{P}^{k''}_{ttt'} } e^{-i\beta  \hat{S}^{k''}_{tt'}} \right) ,
	\end{aligned}
\end{equation}}
has support on all two- and three-qubit excitation operators between qubits sub-systems $z_t^{k'}, \ z_{t'}^{k'}, \ z_{t}^{k''},$ and $z_{t'}^{k''}$---\Cref{eq:robot2} is derived in~\Cref{app:qbm}. 

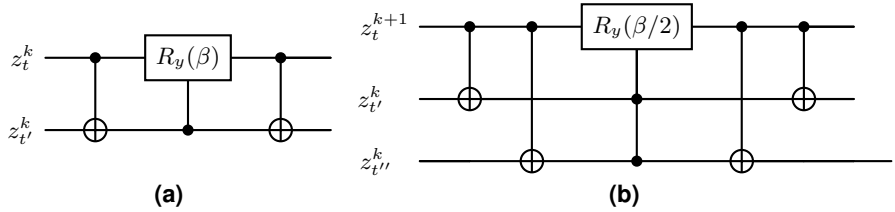
\begin{figure}[h]
	\centering
	\begin{tabular}{c c}
		\begin{quantikz}
			\lstick{${z^k_t}$}  & \ctrl{1} & \gate{ R_y({2} \beta ) }&\ctrl{1}& \qw\\
			\lstick{${z^k_{t'}}$} & \targ{} & \ctrl{-1} &\targ{}& \qw
		\end{quantikz} & 	\begin{quantikz}
			\lstick{${z^{k+1}_t}$}  & \ctrl{1} & \ctrl{2} & \gate{ R_y( {2\beta}) } & \ctrl{2} &\ctrl{1} & \qw\\
			\lstick{$z^{k}_{t'} \quad $} 	  & \targ{}  & \qw	& \ctrl{-1} 	& \qw &\targ{} &\qw\\
			\lstick{$z^{k}_{t''} \ \ $} 	  & \qw      &\targ{}	& \ctrl{-2}& \targ{}& \qw &\qw &
		\end{quantikz}\\
		\textbf{\fontfamily{phv}\selectfont (a)} & \textbf{\fontfamily{phv}\selectfont (b)}
	\end{tabular}
	\caption{Quantum circuit design for (a) {$e^{-i\beta  \hat{S}^{k}_{tt'}}$} and (b) {$e^{-i\beta  \hat{P}^{k}_{tt't''}}$}.} \label{fig:digiEABEAAB}
\end{figure}

\begin{figure}[h]
	\begin{center}
		\includegraphics{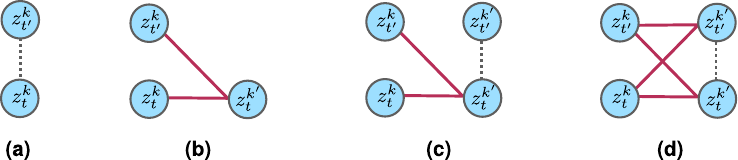}
		\caption{Graphs representing the composition of qubit excitations. The nodes represent the qubits acted upon non-trivially. The edges represent support of the resultant unitary on the generators of two- (dash) and three-qubit excitations (solid). (a) The unitary {$e^{-i\beta  \hat{S}^{k}_{tt'}}$} generated by the two-qubit excitation. (b) The unitary {$e^{-i\beta  \hat{P}^{k}_{ttt'}}$} generated by the three-qubit excitation. (c) The union of the support of the unitaries {$e^{-i\beta  \hat{S}^{k'}_{tt'}}, \ e^{-i\beta  \hat{P}^{k'}_{ttt'} },$} and {$ e^{-i\beta  \hat{S}^{k'}_{tt'}}$} generated by two- and three-qubit excitations. (d) The composition {$e^{-i\beta  \hat{S}^{k'}_{tt'}} e^{-i\beta  \hat{P}^{k'}_{ttt'} } e^{-i\beta  \hat{S}^{k'}_{tt'}}$} also includes support on {$\hat{P}^{k'}_{t't't}$} due to the quantum bridge---see~\Cref{app:three-qubitex}.
		}\label{fig:SPop}
	\end{center}
\end{figure}

By extending this stacked structure consisting of four layers, we can construct a circuit of constant depth with support on all the generators of two- and three-qubit excitations that correspond to valid binary arithmetic operations between assets $t$ and $t'$:
{
\begin{equation}\label{eq:QE23}
	\begin{aligned}
		 \hat{H}_{tt'}(\beta) =  \tilde{S}_{tt'}(\beta) \tilde{P}^e_{tt'}(\beta) \tilde{P}^o_{tt'}(\beta) \tilde{S}_{tt'}(\beta), 
	\end{aligned}
\end{equation}}
where
{\begin{equation}\label{eq:QE2}
	\begin{aligned}
		 & \tilde{S}_{tt'}(\beta) = \bigotimes_{k=1}^l e^{-i\beta \hat{S}^k_{tt'}},      \\
		 & \tilde{P}^e_{tt'}(\beta) = \bigotimes_{k_1 \in \mathcal{K}_1} e^{-i\beta \hat{P}^{k_1}_{ttt'}}, \ \mathcal{K}_1 = \{(2c_1)\bmod l, \ c_1 = 1,2,\dots,\floor*{l/2}\}, \\
		 & \tilde{P}^o_{tt'}(\beta) = \bigotimes_{k_2 \in \mathcal{K}_2 } e^{-i\beta \hat{P}^{k_1}_{ttt'}}, \ \mathcal{K}_2 = \{(2c_2-1)\bmod l, \ c_2 = 1,2,\dots,\ceil*{l/2}\},
	\end{aligned}
\end{equation}}
see \Cref{app:qbm} for details. The graph of the support of $\hat{H}_{tt'}(\beta)$ on the generators of two- and three-qubit excitations is as shown in \Cref{fig:U_{AB}}. 
\begin{figure}[h]
	\begin{centering}
		\includegraphics{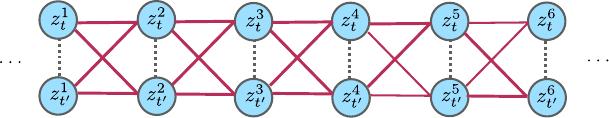}
		\caption{A graph representing the composition $\hat{H}_{tt'}(\beta)$ of qubit excitations. The nodes represent the qubits acted upon non-trivially. The edges represent the support of the resultant unitary on the generators of two- (dash) and three-qubit excitations (solid). 
  } \label{fig:U_{AB}}
	\end{centering}
\end{figure}

We can then define a hard{-constraint} mixing operator as
\begin{equation}\label{mix:def}
	\begin{aligned}
		\hat{M}_L(\beta) = \prod_{d=1}^{L}\prod_{tt'} \hat{H}_{tt'} (\beta),  \quad |t-t'|\bmod n = d, \ 1 \leq L\leq\floor*{n/2},
	\end{aligned}
\end{equation}
where $L$ represents the maximum distance between assets. Assuming all-to-all qubit connectivity the depth complexity is $\mathcal{O}(L)$. 
If we set $L=1$, a nearest-neighbour hard{-constraint} mixing operator is 
\begin{equation}\label{mix:nn}
	\begin{aligned}
		\hat{M}_1(\beta) =  \prod_{tt'} \hat{H}_{tt'} (\beta), \quad |t-t'| \bmod n =1.
	\end{aligned}
\end{equation} 
When choosing $L=\floor*{n/2}$, the mixing operator will contain qubit excitations between all pairs of assets. In the following sections, unless stated, we use $\hat{M}_1(\beta)$ as the mixing operator due to its better performance in the optimisation (see~\Cref{app:simumix}) and the minimal use of quantum gates. 
Further, by repeatedly applying $\hat{M}_L(\beta)$ interleaved with cost operators in the QAOA ansatz, we generate \textit{long-exchange}, \textit{-carry} and \textit{-borrow} operations between assets through transporting excitation from neighbour to neighbour.

{Finally, in~\Cref{app:mix_compare_parity_xyy}, we perform preliminary simulations to compare our mixing operator to the a mixer with the same structure but with parity and XYY mixers substituted for the two- and three-qubit excitation operators, respectively. We find our qubit excitation mixing operator performs best in portfolio optimisation. Thus, we utilise our qubit excitation mixing operator throughout this article.}

\section{Parameterised circuit optimisation} \label{sec:Simulation}
This section benchmarks two classical parameter optimisers, dual annealing (DA) {\cite{xiangGeneralizedSimulatedAnnealing1997}} and constrained optimisation by linear approximation (COBYLA)~\cite{powellDirectSearchOptimization1994}, and two layerwise optimisation methods, \textit{frozen-layerwise} and \textit{unfrozen-layerwise} for our QAOA ansatz.

Optimisation of our ansatz is governed primarily by two hyperparameters which we will consider in turn below: The number of shots $N_m$ used to calculate expectation values and the maximum number of expectation values that can be estimated by the optimiser ${\mathcal{J}}$.

As with most QAOA applications, our simulations use the estimation of an expectation value {$\langle \hat{C} \rangle_{N_m}$} using $N_m$ (optimisation) shots as the objective function. We observe in~\Cref{fig:cbcaOpt} that insufficient shots will lead to a sizeable stochastic noise and an unstable optimisation.  We use $N_M\ge N_m$ (post-optimisation) shots to obtain a more precise estimation of the expectation value of a post-optimised quantum circuit. A further discussion about the effect of different $N_m$ on optimisation is given in \Cref{sec:paraopt}. 

From our simulations, we also find classical optimisers sometimes fail to identify convergence (see left column of~\Cref{fig:cbcaOpt}) due to the poor loss landscape of the quantum circuit and stochastic measurement noise. Thus, employing an ${{\mathcal{J}}}$ is important to limit the overuse of resources. Finally, we denote the total number of objective functions assessed during the optimisation as $N_f$---which is determined by $N_m$ and ${\mathcal{J}}$. 

To evaluate the efficacy of a QAOA after optimisation, a common metric is the mean value approximation ratio:
{\begin{equation}\label{eq:appr_mean}
	\begin{aligned}
		\alpha_{\mean} = \frac{\langle \hat{C}\rangle_{N_M} - f_{\min} }{f_{\max} - f_{\min}},
	\end{aligned}
\end{equation}}
where {$\langle \hat{C} \rangle_{N_M}$} is defined in \Cref{appFp}, $f_{\min}$ and $f_{\max}$ are the global minimal and maximal values of the cost function [~\Cref{qdp}], respectively. $\alpha_{\mean}$ is normalised within the range $[0,1]$. Unlike the approximation ratio in the maxcut problem, our approximation ratio is better when it's lower.
However, the portfolio strategy with the global optimum that a solver can find is more useful for the DGMVP model, so we also introduce the minimal value approximation ratio:
{\begin{equation}\label{eq:appr_min}
	\begin{aligned}
		\alpha_{\min} = \frac{ \min_{i\in\left[1,N_M\right]}\{ \langle \bm{w}_i|\hat{C}|\bm{w}_i \rangle \}- f_{\min} }{f_{\max} - f_{\min}},
	\end{aligned}
\end{equation}}
where the $\bm{w}_i$ is the $i$th measured weight vector.

\subsection{Classical optimisers}\label{sec:paraopt}
{COBYLA is a derivative-free local optimisation method to approximate constrained optimisation problems with linear programming problems iteratively. Benchmarks demonstrate that COBYLA can converge to lower energies faster than other local optimisers when optimising variational quantum algorithms in the absence of noise~\cite{pellow-jarmanComparisonVariousClassical2021,Lavrijsen2020}. While COBLYA is a good local optimiser, we wanted to include a global optimiser. DA is a stochastic global optimisation method that couples a local search with simulated annealing. We did preliminary simulations and found DA to perform well from a range of optimisers. A more thorough comparison with other optimisers is left for the future work.}

In \Cref{fig:cbcaOpt}, we use the COBYLA and DA to optimise one ansatz parameter for different numbers of measurement shots $N_m$ during the optimisation, fixing the DGMVP instance and keeping other ansatz parameters the same. The simulation demonstrates the difficulty of classical optimisers when optimising the loss landscape in the presence of stochastic noise~\cite{gokhaleMeasurementCostVariational2020}. This is particularly evident for COBYLA (see left-hand column of~\Cref{fig:cbcaOpt}): as the number of optimisation shots $N_m$ decreases, COBYLA starts to terminate the optimisation at a stochastically generated position as opposed to a local minimum. Although a sufficiently large $N_m$ can be helpful to the optimisation, it will also lead to excessive cost function accesses, diminishing the advantage of quantum algorithms.
Fortunately, we observe a better result for DA (see right-hand column of~\Cref{fig:cbcaOpt}): a stable performance in finding a point near the global minimum as $N_m$ is reduced with a constant number of expectation value estimations.
Notably, even when decreasing the optimisation shots drastically, as shown in the bottom-right panel of \Cref{fig:cbcaOpt}, DA still stops at a region close to the optimal parameter.
However, DA will fail to identify convergence at the end of optimisation due to stochastic noise, and an ${\mathcal{J}}$ should be set to avoid the over-optimisation cost. 

\begin{figure}[ht] 
	\centering
	\includegraphics{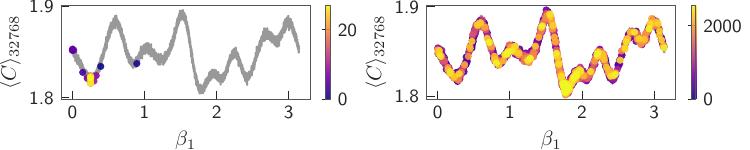}

	\vspace{0.2cm}

	\includegraphics{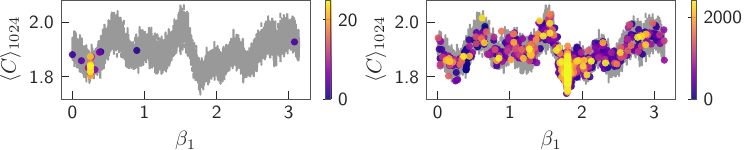}

	\vspace{0.2cm}

	\includegraphics{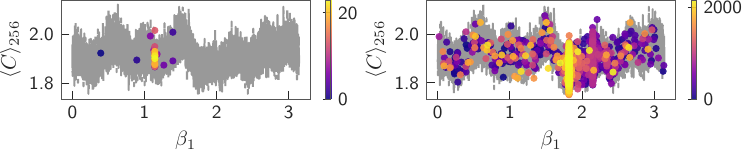} 

	\vspace{0.2cm}

	\includegraphics{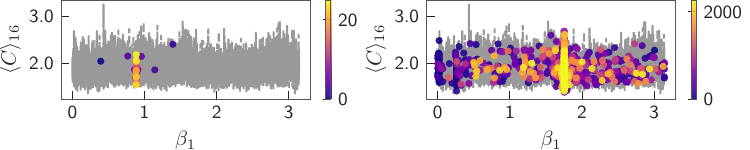}

	\caption[The optimisation simulation]{
        The performance of COBYLA (left column) and DA (right column) when optimising $\beta_1$ with other parameters fixed. These simulations use $n=4$ assets, an asset encoding binary block length of $l=4$, a QAOA circuit with p=10 layers and $\hat{M}_2(\beta)$ as the mixing operator, and $|\boldsymbol{w}\rangle_{mb}$ as the initial state.
		The loss landscape based on sampling from the cost function of $\beta_1$ with stochastic noise is plotted as grey lines. For DA, we use ${\mathcal{J}}=2000$ as the maximal number of expectation value estimations. The level of stochastic noise increases down the page. That is the number of shots utilised for sampling is $N_m$=$2^{15}$, $2^{10}$, $2^{8}$ and $2^{4}$ from the top row to the bottom. The overlayed coloured points indicate the values of the cost function sampled by the optimiser at each step during the optimisation. The colour bar on the right indicates the optimisation steps.
	} \label{fig:cbcaOpt}
\end{figure}

Furthermore, a set of simulations in~\Cref{app:ccoda} shows the advantages of using DA over COBYLA to optimise multilayered ansätze in the presence of stochastic measurement errors. In~\Cref{app:dafuther}, we also consider DA for a range of ${\mathcal{J}}$ and $N_m$ and determine a reasonable choice as ${\mathcal{J}}=2000$ and $N_m =16$, which we employ in~\Cref{sec:analysis}.

\subsection{Layerwise optimisation}\label{sec:layerwise}
Layerwise optimisation is becoming a popular optimisation method for the QAOA and indicates faster convergence in many simulations~\cite{skolikLayerwiseLearningQuantum2021,PhysRevResearch.4.033029,PhysRevA.104.L030401,PhysRevA.109.052406}. 
There are many variant versions of the layerwise optimisation methods, and they can be categorised as {\textit{frozen-layerwise}} and {\textit{unfrozen-layerwise}} for our ansatz.
For a fixed $p>1$, both methods start with optimising the $p'=1$ circuit with random initial parameters and store the optimised parameters as $\bm{\gamma}', \bm{\beta}'$. Then, we add subsequent layers by setting the previous layers' parameters as $\bm{\gamma}', \bm{\beta}'$ and initialising the new layer's parameters as zero (an identity layer) to avoid detrimenting the prior optimisation. On the one hand, in frozen-layerwise optimisation, we fix parameters from before the new layer is added and only optimise the parameters in the new layer. On the other hand, in unfrozen-layerwise optimisation, we optimise all of the layers' parameters together. For both methods, a well-optimised $p$ layer quantum circuit can then be iteratively constructed by adding a new layer and optimising. When using DA as the layerwise optimiser, we limit the number of expectation value estimations used during optimisation after adding each new layer to {$\mathcal{J}'$}. Thus, after adding a total of $p$ layers, the maximum number of expectation value estimations is {$\mathcal{J}=\mathcal{J}'p$}.
{Additionally, we use the \textit{fixed-$p$ ansatz} for benchmarking, referring to optimisation of all parameters in an ansatz at fixed $p$ layers.}

In~\Cref{fig:layerwise}, we perform a comparison of the unfrozen-layerwise, the frozen-layerwise, and the fix{ed-$p$} ansatz optimisation methods with COBYLA and DA to a fixed DGMVP instance with a fixed number of measurement shots $N_m=1024$ during the optimisation. The simulation shows that unfrozen-layerwise optimisation, in general, outperforms frozen-layerwise in both $\alpha_{\mean}$ and $\alpha_{\min}$. 
The value of $\alpha_{\mean}$ after optimisation is roughly constant after $p=2$ for both layerwise methods with COBYLA and DA. In the COBYLA simulations [see \Cref{fig:layerwise}(a)], we observe a lower $\alpha_{\mean}$ for both layerwise methods compared to the fix{ed-$p$} ansatz approach. 
Furthermore, for the DA simulations, the unfrozen-layerwise optimisation in \Cref{fig:layerwise}(c) uses the same number of expectation value estimations in total but achieves a better result than the fix{ed-$p$} ansatz approach for circuits with larger $p$. {However, $\alpha_{\mean}$ is used only during the optimisation. Our algorithm aims to find the bit string with minimum cost. Thus, the quality of the output of our algorithm is the minimum cost measured. Therefore, $\alpha_{\min}$ better reflects the performance of our algorithm. Despite the plateau in $\alpha_{\mean}$, we find the median of $\alpha_{\min}$ is pinned at zero, in \Cref{fig:layerwise}(b) and (d),} for all methods but frozen-layerwise with DA at almost all $p$---at least 50\% of the time the true global minimum is found. Thus, we plot the 80\% percentile to indicate the tail of the distribution. When using COBYLA, we observe that the fix{ed-$p$} ansatz approach gives the tightest distribution for the larger $p$ circuits{---the output of the algorithm improves with the number of layers}. Further, we find that the unfrozen-layerwise approach using DA returns the global minimum with a probability greater than 80\% for $p\geq3$ and outperforms all other methods presented. 
Compared to optimising the $p$ layer circuit directly as in the fix{ed-$p$} ansatz approach, the unfrozen-layerwise optimisation also uses a shorter circuit until adding the last layer which reduces quantum resource usage.
The layerwise methods may also be more noise-robust by minimising the noise induced by finite coherence times during the early stages of optimisation. 
 
\begin{figure}[h]
	\centering
	\includegraphics{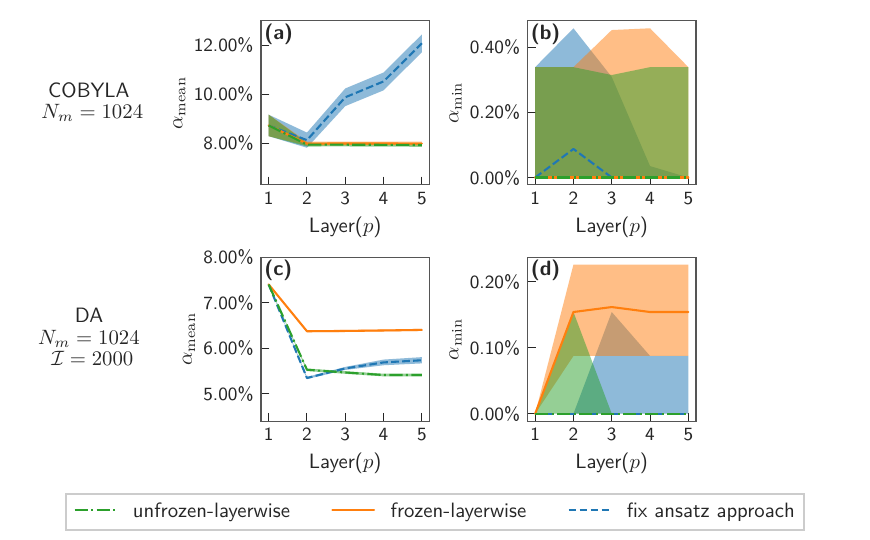}
	\caption[The layerwise optimisation]{A comparison of optimisation methods: (a) and (b) use COBYLA as the optimiser with $N_m = 1024$; (c) and (d) use DA as the optimiser with ${\mathcal{J}} = 2000$ and $ N_m = 1024$.   
	The lines in (a) and (c) are the mean of $\alpha_{\mean}$ over stochastic simulations with 100 randomly generated initial ansatz parameters for $p=1$. The lines in (b) and (d) are the median of $\alpha_{\min}$. Shaded regions take percentiles 80\% of data. These simulations use the mixing operator $\hat{M}_1(\beta)$ and the initial state $|\boldsymbol{w}\rangle_{mb}$ with $n=l=4$. The final post-optimisation expectation value measurement shots of $N_M = 2^{16} = 65536$.
	}
	\label{fig:layerwise}
\end{figure}

\section{Scaling results} \label{sec:analysis}
To investigate the performance of our algorithm in solving the DGMVP model, we conduct stochastic simulations of models randomly selected from the financial market with varying asset numbers $n$ and binary block lengths $l$. {As we target the practical advantages of solving~\Cref{qdp}, we are interested in the probability of our algorithm generating portfolio with low variances. Thus, we provide a comprehensive analysis of distribution generated by the post-optimisation quantum circuit.}  
However, the stochastic noise in estimating the final expectation value {$\langle \hat{C}\rangle_{N_M}$} of the post-optimisation ansatz will increase with problem size if we fix the number of shots $N_M$. Thus, to gain a precise description of the optimised ansatz, we rebuild the parameterised quantum circuit using a statevector representation and calculate the exact expectation value {$\langle \hat{C}\rangle$}. This allows us to benchmark the scaling of our quantum algorithm's performance with a range of metrics independent of $N_M$: the approximation ratios within the lower $k$ percentile $\alpha^{k\%}_{\mean}$ ($\alpha^{100\%}_{\mean} = \alpha_{\mean}$); the minimal value approximation ratio $\alpha_{\min}$; the probability of measuring the minimal value $\mathrm{P}_{\min}$; the probability of measuring global minimum $\mathrm{P}_{gm}$. The results are presented in \Cref{fig:anaylysis_scale_l_n_approx_mb,fig:anaylysis_scale_l_n_approx_ws,fig:anaylysis_scale_l_n}---for simulation details see~\Cref{app:sim}. 

\begin{figure}[h]
	\centering
	\includegraphics{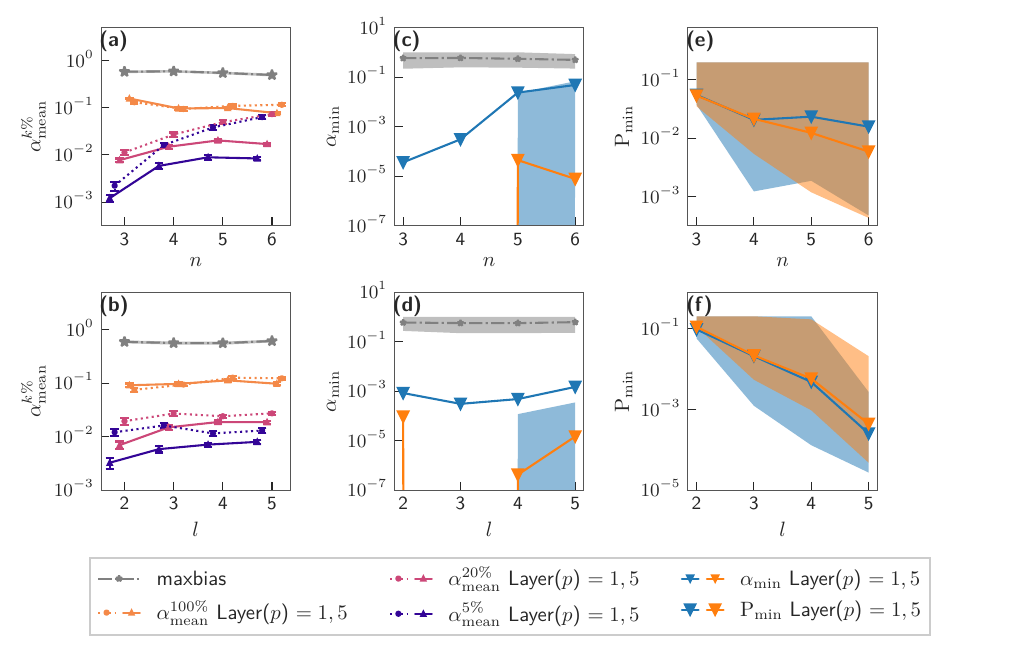}
	\caption[The scaling analysis in n and l approximate ratio]{The $\alpha^{k\%}_{\mean}$ (a) and (b), $\alpha_{\min}$ (c) and (d) and $\mathrm{P}_{\min}$ (e) and (f) of the stochastic optimisation over 100 randomly generated DGMVP instances as a function of $n$ with $l=3$ (the top three panels) and $l$ with $n=4$ (the bottom three panels). The simulations use the initial state $|\bm{w}\rangle_{mb}$ and the mixing operator $\hat{M}_1(\beta)$, and the initial ansatz parameters are randomly generated. The grey lines are the approximation ratio of maxbias {classical benchmark} portfolios (also considered initial states). The coloured lines are the average data of $\alpha^{k\%}_{\mean}$, $\alpha_{\min}$ and $\mathrm{P}_{\min}$ of the post-optimised ansatz. The shade regions and error bars in (a) and (b) represent the standard error of the data. 
  The shaded regions in (c) and (d) are 25\%-75\% percentiles of the data, and the shaded regions in (e) and (f) are 50\%-100\% percentiles of the data. The missing data points in (c) and (d) for $\alpha_{\min}$ correspond to values equal to $0$. 
 The classical optimiser is DA with ${\mathcal{J}} = 2000$ and $N_m = 16$, and the optimisation method is unfrozen-layerwise optimisation.}
\label{fig:anaylysis_scale_l_n_approx_mb}
\end{figure}

In \Cref{fig:anaylysis_scale_l_n_approx_mb,fig:anaylysis_scale_l_n_approx_ws,fig:anaylysis_scale_l_n}, we compare {two portfolios with the portfolio generated by} our quantum algorithm. 
We use the maxbias and the warm-started portfolio as both classical strategies and our initial states (see~\Cref{sec:initialstate}). For the warm-started portfolio, we post-select {out the instances that are not the exact global minimum solutions in order to compare the level of improvement---see~\Cref{tab:selec_ratio} for the percentage of instances retained.

\begin{table}[!ht]
    {\centering
    \setlength{\tabcolsep}{6pt}
    \begin{tabular}{lrrrrrrr}
\toprule
& $n=3$ & $n=4$ & $n=5$ & $n=6$ & $n=4$ & $n=4$ & $n=4$ \\ 
& $l=3$ & $l=3$ & $l=3$ & $l=3$ & $l=2$ & $l=4$ & $l=5$ \\
\cmidrule{2-8}
Retention percentage & $8.99\%$ & $17.48\%$ & $24.54\%$ & $30.20\%$ & $14.27\%$ &  $14.85\%$ & $14.46\%$ \\ 
\bottomrule
\end{tabular}
\caption{The percentage of instances retained during post-selection of the warm-started portfolio on random instances. We retain instances for which the warm-started portfolio is not the {optimal} DGMVP solution. $n$ is the number of assets, and $l$ is the asset encoding binary block length.}
\label{tab:selec_ratio}
}
\end{table}
} 
For these two {portfolios}, we observe that the warm-started {{classical benchmark} portfolio} outputs better results by orders of magnitude compared to the maxbias {classical benchmark portfolio}. For our quantum algorithm, we use DA as the optimiser and unfrozen-layerwise optimisation method with $\hat{M}_1(\beta)$. As the computational cost of optimising for the hyperparameters, $N_m$ and ${\mathcal{J}}$, for all $n$ and $l$ configurations is high, we fix $N_m = 16$ and ${\mathcal{J}} =2000$ as motivated by~\Cref{app:dafuther}. We justify fixing these values for all $n$ and $l$ as follows: Yanakiev \textit{et al}. \cite{PhysRevA.109.032420} demonstrates in appendix D that the number of shots required to estimate the approximation ratio to a fixed precision is independent of $n$ and $l$. Further, we assume the number of optimisation steps to reach a given accuracy depends only on the number of parameters and hence the number of layers. Thus, we choose to fix the total number of shots used for optimisation for each layer.

In \Cref{fig:anaylysis_scale_l_n_approx_mb}(a) and (b), we observe our quantum algorithm achieves a significantly lower $\alpha^{k\%}_{\mean}$ for all tested $k$ and $\alpha_{\min}$ compared to the maxbias {classical benchmark portfolio} as a function of $n$ and $l$. Additionally, we find $\alpha^{k\%}_{\mean}$ and $\alpha_{\min}$ improve with increasing $p$ in the maxbias{-initialised} simulations. In particular, for the small $l$ and $n$ configurations, the post-optimised quantum circuit can measure the {optimal} DGMVP solution with a relatively high probability---see~\Cref{fig:anaylysis_scale_l_n_approx_mb}(c) and (d). We observe $\alpha^{k\%}_{\mean}$ appears to saturate [see ~\Cref{fig:anaylysis_scale_l_n_approx_mb}(a) and (b)] and $\mathrm{P}_{\min}$ decreases [see ~\Cref{fig:anaylysis_scale_l_n_approx_mb}(e) and (f)] as a function of $n$ or $l$. This could be because we use fixed values for $N_m$ and ${\mathcal{J}}$ for all the $n$ and $l$ configurations.  

When using the warm-started initial state, we find $\alpha^{100\%}_{\mean}$ is worse than the {approximation ratio of the warm-started {classical benchmark} portfolio itself. This is potentially }due to stochastic noise at finite $N_m$ and ${\mathcal{J}}$ during optimisation---see~\Cref{fig:anaylysis_scale_l_n_approx_ws}(a) and (b). {We} find that{, while the optimisation increases the mean,} the distribution {skews} towards the global minimum: {T}he $\alpha^{5\%}_{\mean}$ and $\alpha^{20\%}_{\mean}$ curves lie close to the warm-started {classical benchmark} approximation ratio. This is likely due to an exacerbation of the stochastic noise from additional states being mixed into the statevector. {Additionally, increasing $p$ increases $\alpha^{k\%}_{\mean}$.} {However, it is $\alpha_{\min}$ that characterises the performance of our algorithm.} The post-optimisation $\alpha_{\min}$ for the warm-started initial state [see~\Cref{fig:anaylysis_scale_l_n_approx_ws}(c) and (d)] is smaller than the warm-started {classical benchmark} portfolio. {Also, while the mean of $\alpha_{\min}$ over 100 instances does not change with $p$, the distribution of $\alpha_{\min}$ skews towards the global minimum.} {Further, t}he corresponding $\mathrm{P}_{\min}$ [see~\Cref{fig:anaylysis_scale_l_n_approx_ws}(e) and (f)] is almost constant as a function of $n$ and $l$ with fixed $N_m$ and ${\mathcal{J}}$ throughout the optimisation.

\begin{figure}[h]
 \centering
 \includegraphics{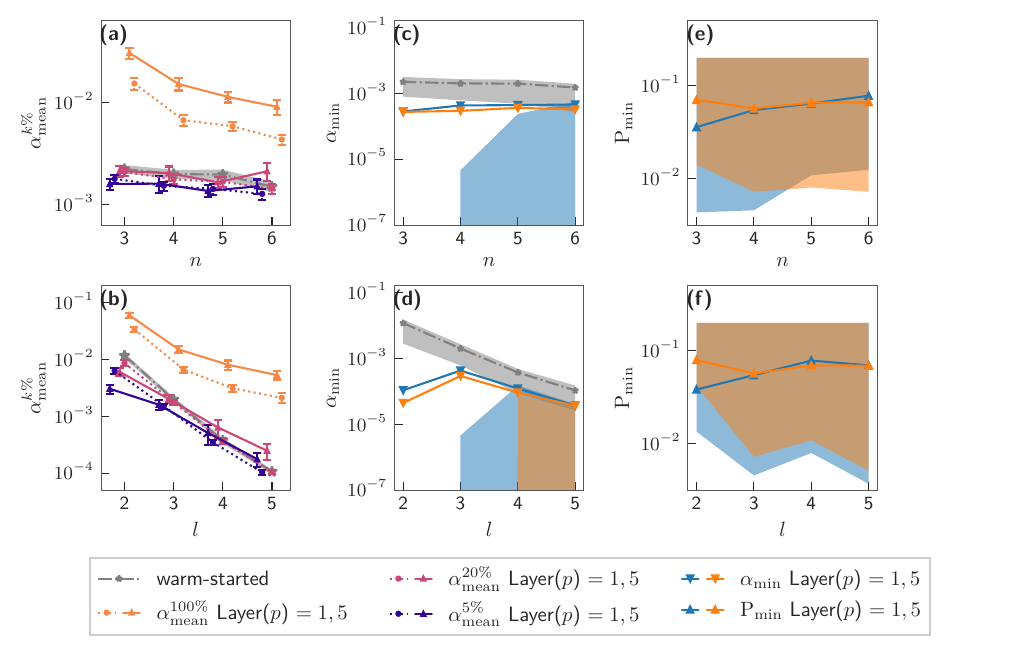}
 \caption[The scaling analysis in n and l approximate ratio]{The $\alpha^{k\%}_{\mean}$ (a) and (b), $\alpha_{\min}$ (c) and (d) and $\mathrm{P}_{\min}$ (e) and (f) of the stochastic optimisation over 100 randomly generated DGMVP instances as a function of $n$ with $l=3$ (the top three panels) and $l$ with $n=4$ (the bottom three panels). The simulations use the initial state $|\bm{w}\rangle_{ws}$ and the mixing operator $\hat{M}_1(\beta)$, and the initial ansatz parameters are randomly generated. The grey lines are the approximation ratio of {warm-started classical benchmark} portfolios (also considered initial states). The coloured lines are the average data of $\alpha^{k\%}_{\mean}$, $\alpha_{\min}$ and $\mathrm{P}_{\min}$ of the post-optimised ansatz. The shade regions and error bars in (a) and (b) represent the standard error of the data. 
  The shaded regions in (c) and (d) are 25\%-75\% percentiles of the data, and the shaded regions in (e) and (f) are 50\%-100\% percentiles of the data.
 The classical optimiser is DA with ${\mathcal{J}} = 2000$ and $N_m = 16$, and the optimisation method is unfrozen-layerwise optimisation.}
 \label{fig:anaylysis_scale_l_n_approx_ws}
\end{figure}

A global minimum is more important for investment because it represents the lowest-risk portfolio. Thus, in~\Cref{fig:anaylysis_scale_l_n}, we plot the inverse of the probability of measuring the global minimum $1/\mathrm{P}_{gm}$. The number of measurements we require to find the true global minimum with some fixed success probability $\mathrm{P}_{s}$ is $N_M=\log(1-\mathrm{P}_{s})/\log(1-\mathrm{P}_{gm})$. We benchmark this against the probability of measuring the global minimum when drawing uniformly random samples from both the constrained and unconstrained feasible region.
The figure shows that the optimised quantum circuit has a significantly lower $1/\mathrm{P}_{gm}$ than the constrained feasible region size, meaning the circuit can find the {optimal} DGMVP solution with much higher probability for a single shot compared to a constrained {uniform} random sampling method. 

We also fit the inverse of the mean of $\mathrm{P}_{gm}$ (the crosses in~\Cref{fig:anaylysis_scale_l_n}) with the function $ 1/\bar{\mathrm{P}}_{gm}  = a \cdot B(n,l)^b$ where $B(n,l)$ is chosen as the constrained feasible region size (describe in~\Cref{appendix:scaling}). We find $b= {0.76 \pm 0.30}$ for the warm-started initial state, and $b = {0.51 \pm 0.09}$ for the maxbias initial state---see \Cref{fig:anaylysis_scale_l_n}(a). The ratio of the exponents of $1/\bar{\mathrm{P}}_{gm}$ for the constrained {uniform} sampling method to our quantum algorithm with warm-started and maxbias initial states are $1/0.76 \approx 1.32$ and $1/ 0.51 \approx 1.96$, respectively. This indicates our quantum algorithm has a more favourable power-law scaling with $n$.
Similarly, we find $b={0.66 \pm 0.01}$ and $b={1.11 \pm 0.05}$ for the warm-started and maxbias initial states, respectively---see \Cref{fig:anaylysis_scale_l_n}(b). Once again, we can consider the ratio of the exponents: $1/0.66\approx 1.52$ and $1/1.11 \approx 0.90$ for the warm-started and maxbias initial states, respectively. Therefore, the scaling with $l$ of our quantum algorithm appears to be more favourable than constrained {uniform} sampling when using the warm-started initial state. The scalings in~\Cref{fig:anaylysis_scale_l_n} indicate our quantum algorithm utilising the warm-started initial state may asymptotically obtain a reduction in the number of samples required to find the {optimal} DGMVP solution over {constrained uniform sampling method}.
{However, the constrained uniform sampling method only provides a lower bound for the classical sampling complexity. We leave comparisons to improved classical algorithms for future work.}

\begin{figure}[h]
	\centering
	\includegraphics{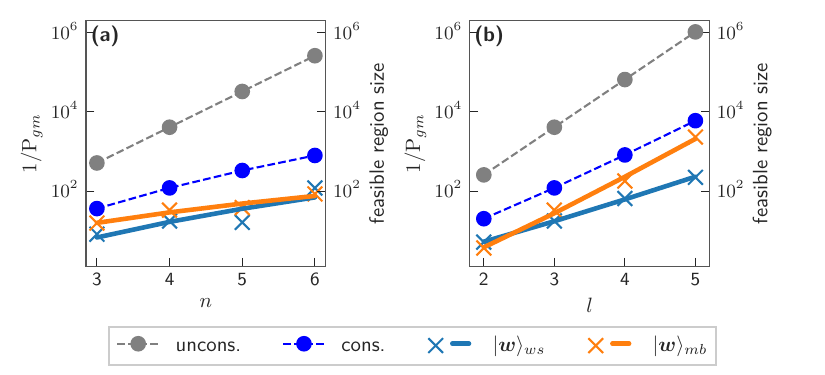}
	\caption[The scaling analysis in n and l]{
	 The inverse of the probability to measure the {optimal} DGMVP solution ($1/\mathrm{P}_{gm}$) and the feasible region scaling of the maxbias $|\boldsymbol{w}\rangle_{mb}$ and warm-started $|\boldsymbol{w}\rangle_{ws}$ initial states as a function of (a) $n$ with $l=3$, (b) $l$ with $n=4$. The maxbias optimisation data is the same as in \Cref{fig:anaylysis_scale_l_n_approx_mb}, and the warm-started optimisation data is the same as in \Cref{fig:anaylysis_scale_l_n_approx_ws}.
    In both panels, the left vertical axes are the $1/\mathrm{P}_{gm}$ for the crosses and solid lines. Each cross is the average of the measured data, and the solid lines with the same colour fit the cross points with a function of $ 1/\bar{\mathrm{P}}_{gm} = a \cdot B(n,l)^b$ where $B(n,l)$ is the constraint feasible region size function. 
    The right vertical axes denote the feasible region size as in \Cref{Eq:Assign3} for the grey (unconstraint feasible region) and the blue (constraint feasible region) dash line and round solid dots. }
	\label{fig:anaylysis_scale_l_n}
\end{figure}

\section{Quantum noise} \label{sec:noise}
This section compares the performance of our algorithm in the presence of thermal-relaxation noise with and without post-selection during optimisation. We utilise Qiskit's thermal-relaxation noise model parameterised by an ensemble of $T_1$ and $T_2$ times drawn from normal distributions such that $T_2\le 2T_1$ with $T_1\sim \mathcal{N}(50\mu s, 10\mu s)$ and $T_2\sim \mathcal{N}(70\mu s, 10\mu s)$~\cite{qiskit2024}.
To evaluate the influence of the quantum noise on our algorithm, we compare two methods: First, we use measurements on a simulated noisy quantum processor to estimate the expected value directly during the optimisation. Second, we use a post-selection method to enforce the constraint in~\Cref{qdp:2a} when estimating the expectation value during the optimisation.

Noise-induced errors can invalidate the constraint in~\Cref{qdp:2a}, which will increase the difficulty of optimisation when we do not apply post-selection. However, even when we apply post-selection, optimisation will still be hampered by noise-induced errors that leave the constraint in~\Cref{qdp:2a} satisfied. Further, the application of post-selection reduces the number of shots that can be utilised for optimisation, which either increases the stochastic noise or the runtime. These effects compound to deteriorate the quality of the optimisation as shown in \Cref{fig:noise}(a) and (b). Finally, the average post-selection probability for both methods during the optimisation is presented in \Cref{fig:noise}(c) and (d), demonstrating roughly an order of magnitude increase in the number of shots required to maintain a fixed level of stochastic noise.

These results corroborate theoretical~\cite{stilckfrancaLimitationsOptimizationAlgorithms2021,PRXQuantum.4.010309} and numerical~\cite{dalton2022variational,PhysRevA.109.042413,PhysRevA.109.032420} results, suggesting only limited quantum advantage can be obtained by utilising variational quantum algorithms unless quantum noise levels are decreased by orders of magnitude.
In a similar manner to stochastic noise, in~\Cref{fig:noise}(a), increasing measurement shots $N_m$ can reduce the influence of quantum noise. In particular, the filtered method shows a small improvement in \Cref{fig:noise}(a) compared to the unfiltered one as $N_m$ is increased.
However, when increasing $p$, the gap of $\alpha_{\mean}$ closes and the probability of post-selection $\mathrm{P}_{ps}$ dropping sharply [see \Cref{fig:noise}(b) and (d)].

\begin{figure}[h]
	\centering
	\includegraphics{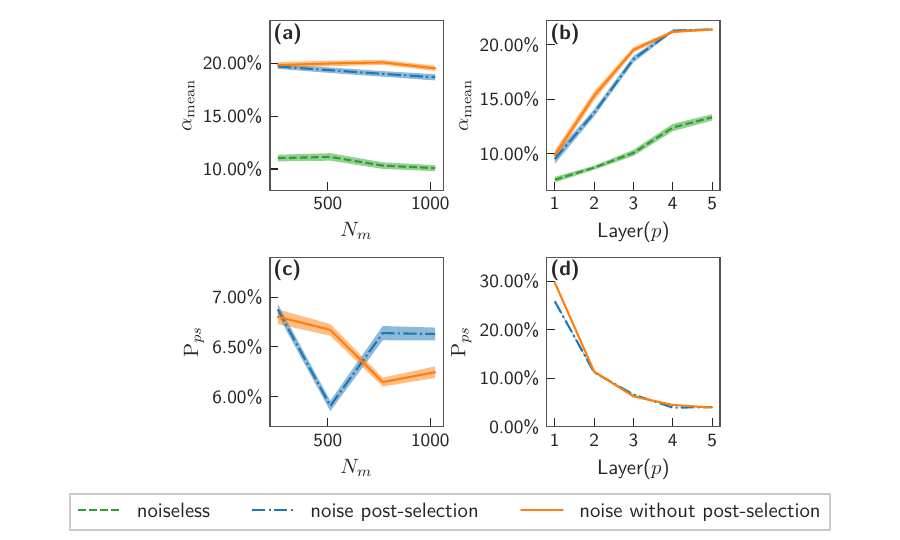}
	\caption[The simulation of noise.]{Performance of our QAOA in the presence of thermal relaxation noise with and without post-selection during the optimisation with optimiser COBYLA and optimisation method fix{ed-$p$} ansatz approach. 
 We present both $\alpha_{\mean}$ [(a) and (b)] and the average probability of post-selection during parameter optimisation $\mathrm{P}_{ps}$ [(c) and (d)] as functions of $N_m$ with ansatz layer $p=3$ [(a) and (c)] and layer $p$ with $N_m=1024$ [(b) and (d)]. The shaded regions in (a), (b), and (d) are the 25\% to 75\% quantile, and the shaded region in (c) is the 10\% to 90\% percentile of the data. These simulations use the mixing operator $\hat{M}_1(\beta)$ and the initial state $|\boldsymbol{w}\rangle_{mb}$ with $n=l=4$. The final expectation value measurement shots $N_M = 2^{16} = 65536$. The results are a summarisation of over 100 simulations with random initial ansatz parameters. 
	}
	\label{fig:noise}
\end{figure}

\section{Conclusion} \label{sec:conclusion}
In this article, we have provided an end-to-end QAOA method to solve the DGMVP model. Further, we analysed the advantages and issues with the QAOA method for financial applications. Based on the DGMVP model, we designed the max-biased, ranked warm-started, approximated equal-weighted, and random-weighted initial states. Our simulations showed the max-biased and warm-started initial states outperform the others. We mapped the DGMVP model to a Hamiltonian {$\hat{C}$} and designed corresponding cost operators {$\hat{U}(\hat{C}, \gamma)$} with depth linear in the number of qubits. We designed a new hard{-constraint} mixing operator $\hat{M}_L(\beta)$ utilising two- and three-qubit excitations with efficient use of quantum gates using a phenomenon we call the quantum bridge. We numerically demonstrated that lower coupling distances between assets in the mixing operator makes the ansatz more readily optimisable.

Through comparative simulations, we found using DA with the limited number of expectation value access is better than one of the best local optimisation methods, COBYLA, even with large stochastic noise. 
We analyse two layerwise optimisation methods, frozen-layerwise, and unfrozen-layerwise, to improve the optimisability. Both methods outperform the fix{ed-$p$} ansatz optimisation approach, but unfrozen-layerwise indicates a better optimisation result in general. 

In the scaling analysis section, we present our quantum algorithm's performance as a function of asset numbers $n$ and block lengths $l$. 
On the one hand, we found our quantum algorithm can decrease the mean and minimal value approximation ratio from the maxbias initial state. 
On the other hand, the mean approximation ratio is not improved by optimisation from the warm-started initial state. However, the distribution is skewed towards the global minimum, which allows an improved solution to be measured with constant probability as a function of $n$ and $l$.
Moreover, further investigation shows that our quantum algorithm utilising the warm-started initial state may asymptotically (in $n$ and $l$) outperform the {constrained uniform sampling} in the number of samples required to find the {optimal} DGMVP solution.

Unfortunately, our quantum algorithm performs poorly in the presence of thermal relaxation noise. The impact can be mitigated by post-selecting out states not within the constrained feasible region. However, post-selection requires significantly more measurements to maintain a constant stochastic noise in the presence of quantum noise. This indicates drastic improvements in error rates or fault-tolerant quantum computation will be required for viable quantum finance applications.

In short, we showed that our quantum approximation algorithm allows us to predict the DGMVP. Our simulations and analyses reveal the potential advantages and disadvantages of implementing quantum algorithms in the financial industry in the future.

\section*{Acknowledgment}
The authors thank Daniel Stilck Fran\c{c}a for valuable discussions on the DGMVP models and scaling figures.
We thank Hanqing Wu and Binren Chen for the discussion of the binary encoding method, Jedrzej Burkat and Jonathan J. Thio for suggestions on mathematical notations, and David R. M. Arvidsson-Shukur and Normann Mertig for giving helpful suggestions and comments. 

\section*{Data availability statement}
The data collected for all the figures and tables in the article are available in reference~\cite{yuanSupportingDataQuantifying2025}.

\section*{Code availability statement}
The data and figures can be reproduced using the code available in references~\cite{yuanQuantumSolverGlobal2026} and~\cite{yuanQuantifyingAdvantagesSourceCode2026}, respectively.

\bibliographystyle{apsrev4-2}
\bibliography{main}

\newpage
\appendix
\crefalias{section}{appendix}

\section{Covariance matrix for portfolio optimisation} \label{app:covariance}
A financial asset $i$ has typically a time series of price data. If we represent it by a random variable $P_i(\bm{T})$, where {$\bm{T} = (T_1,T_2,\dots,T_n)$} is the time series, a $\sigma_{ij}$ can be calculated as
\begin{equation}\label{matrixA2}
	\begin{aligned}
		\sigma _{ij}={\operatorname {cov} (P_i(\bm{T}),P_j(\bm{T}))}
	\end{aligned}
\end{equation}
where $\operatorname {cov}$ is the covariance.
This article generates a portfolio basket among 32 stocks from the NASDAQ and NYSE, dated from "2023-03-01" to "2023-06-30". The price data uses their daily adjust-close price over the same 100 days as raw data.
An example of a basket of 8 selected stocks covariance matrix is shown in \Cref{fig:cov}.

\begin{figure}[ht]
	\centering
	\includegraphics{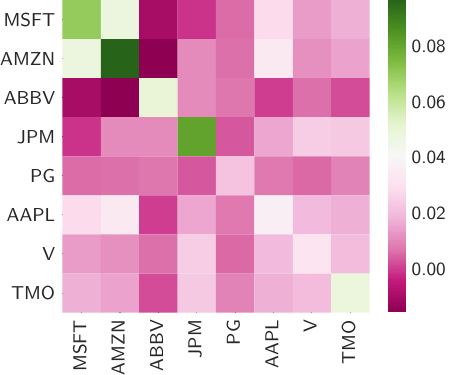}
	\caption{The covariance matrix of 8 selected stocks. MSFT means Microsoft Corp, AMZN means Amazon.com Inc, ABBV means AbbVie Inc, JPM means JPMorgan Chase \& Co, PG means 
Procter \& Gamble Co, APPL means Apple Inc, V means Visa Inc, and TMO means Thermo Fisher Scientific Inc.}\label{fig:cov}
\end{figure}

\section{Feasible region size of the discrete global minimum variance portfolio (DGMVP) with and without the budget constraint}\label{appendix:scaling}
This appendix reveals the method to calculate the feasible region size of the DGMVP with and without the budget constraint, which also represents the total number of portfolio strategies. The feasible region size for the DGMVP [\Cref{qdp}] without the constraint is $2^{ln}$, where $l$ is the number of binary variables for encoding each asset and $n$ is the number of assets.

For the constrained DGMVP as in \Cref{qdp}, the total number of portfolio strategies can be calculated inductively. 
Let $A(n, \mathrm{N})$ be the number of strategies on $n$ assets with $\mathrm{N}$ trading lots. Suppose we spend $m$ trading lots on the $n$th asset then we will have $A(n-1, \mathrm{N}-m)$ remaining strategies on the remaining $n-1$ assets. Therefore, the total number of strategies will be:
\begin{equation}\label{Eq:Assign1}
	\begin{aligned}
		 & A(n,\mathrm{N}) = A(n-1,\mathrm{N}) +A(n-1,\mathrm{N}-1) +\dots + A(n-1,1) +A(n-1,0)
	\end{aligned}
\end{equation}
Inductively, $A(n, \mathrm{N})$ can be calculated as proven in~\Cref{lemma:feas}. 
\begin{lemma}\label{lemma:feas}
The number of strategies $A(n, \mathrm{N})$ of $n$ assets with $\mathrm{N}$ trading lots is
\begin{equation}\label{Eq:Assign2}
	\begin{aligned}
		 & A(n, \mathrm{N}) = \binom{\mathrm{N}+n-1}{n-1}.
	\end{aligned}
\end{equation}
\end{lemma}
\begin{proof}
Suppose~\Cref{Eq:Assign2} holds true for some specific $n$ and all $\mathrm{N}'$, then
\begin{equation}
	\begin{aligned}
		  A(n+1, \mathrm{N}') &= \sum_{m =0}^{\mathrm{N}'}A(n,m)\\
            &= \sum_{m =0}^{\mathrm{N}'} \binom{m+n-1}{n-1} \\
            &= \binom{\mathrm{N}'+n+1-1}{n+1-1},\\
	\end{aligned}
\end{equation}
where the last equality follows from the rules of Pascal's triangle:
\begin{equation}
	\begin{aligned}
		  \binom{j}{k}\equiv \binom{j+1}{k+1}-\binom{j}{k+1}.
	\end{aligned}
\end{equation}
Note that $A(1,\mathrm{N}')=1\equiv\binom{\mathrm{N}'}{0}$ for all $\mathrm{N}'$. Thus, by induction \Cref{Eq:Assign2} must hold for all $n$ and $\mathrm{N}'$.
\end{proof}

When considering a binary encoding, we define a new notation of the number of strategies:
\begin{equation}\label{Eq:Assign3}
	\begin{aligned}
		 & B(n, l) = \binom{2^l+n-2}{n-1},
	\end{aligned}
\end{equation}
where $l$ is the binary block length for encoding each asset. 

In \Cref{fig:anaylysis_scale_l_n}, we observe both the constrained and unconstrained feasible reason sizes grow exponentially. However, the size of the feasible region with the budget constraint grows much slower than the one without constraint. The gaps become more significant on an exponential scale, which indicates our hard{-constraint} mixing operator will exclude drastic amounts of unqualified solutions in a larger $n,l$ configuration.

\section{Initial state extra preparation methods and simulations}\label{app:inistat}
In this appendix, we compare the initial states defined in \Cref{sec:initialstate} and below. To facilitate this comparison, we present both $\alpha_{\mean}$ and $\alpha_{\min}$ as a function of $p$ for each initial state in \Cref{fig:initialCom}. We find distinctly different trends in $\alpha_{\mean}$ and $\alpha_{\min}$: On the one hand for $\alpha_{\mean}$, there is an optimal $p$ before which the stochastic measurement noise causes the optimiser to return increasingly poor values for $\alpha_{\mean}$---for a thorough analysis of stochastic measurement noise see~\Cref{sec:paraopt}. On the other hand $\alpha_{\min}$ generally decreases with increasing $p$ and is impacted less by the stochastic measurement noise. Fortunately, despite optimising with respect to estimated averaged values, it is $\alpha_{\min}$ that defines the quality of our final solution to the DGMVP model. Thus, we conclude that $|\boldsymbol{w}\rangle_{ws}$ performs the best followed by $|\boldsymbol{w}\rangle_{mb}$. Henceforth, we will restrict our investigations to $|\boldsymbol{w}\rangle_{ws}$ and $|\boldsymbol{w}\rangle_{mb}$; not only for their superior performance but also as they represent two disjoint areas for potential quantum advantage: $|\boldsymbol{w}\rangle_{ws}$ represents a \textit{prior-knowledge} initialisation as a good classical approximation obtainable in polynomial time and allows us to probe the potential for improved quantum-enabled approximations in the non-asymptotic regime. Alternatively, $|\boldsymbol{w}\rangle_{mb}$ represents a \textit{zero-knowledge} initialisation allowing us to investigate possible quantum-enabled accelerations in obtaining an approximation at least as good as $|\boldsymbol{w}\rangle_{ws}$.

\textit{The approximated equal-weighted state $|\boldsymbol{w}\rangle_{ew}$:---}This state represents an approximation to an equal-weighted portfolio strategy (50/50 portfolio strategy) and consists of buying each asset in equal weight w.r.t. the total capital. Researchers often use an equal-weighted portfolio strategy as a benchmark for other strategies~\cite{malladiEqualweightedStrategyWhy2017}. An approximated equal-weighted strategy of the DGMVP can be found by iteratively buying the unit trading lot (binary precision $a$) from the first asset to the end. For a loop of buying one trading lot for $n$ assets, the total spent is $na$. The maximal repetition of completing this buying loop is at most $\floor*{1/(n a)}$ times because of the budget constraint. Then, we assign remaining budgets to buy one trading lot $a$ sequentially from the first asset to the $p$th until no money left and get an equation as
\begin{equation}\label{eq:eqweigt}
	\begin{aligned}
		 & 1- \floor*{\frac{1}{n a}} n a-p a=0.
	\end{aligned}
\end{equation}
An approximated equal-weighted state $|\boldsymbol{w}\rangle_{ew}$ is to prepare the asset weights as 
{\begin{equation}\label{eq:ew}
	w_{i} = \left\{\begin{array}{l}
		a\tilde{z}+a, i\leq p\\
		a\tilde{z}, i>p
		\end{array}\right.,
\end{equation}}
where {$\tilde{z} = \floor*{1/(n a)}$}.

\textit{The random-weighted state $|\boldsymbol{w}\rangle_{rd}$:---} This state is prepared by generating a uniformly random set of weights satisfying the budget constraint and the binary block encoding. 
 
\begin{figure}[h]
	\centering	\includegraphics{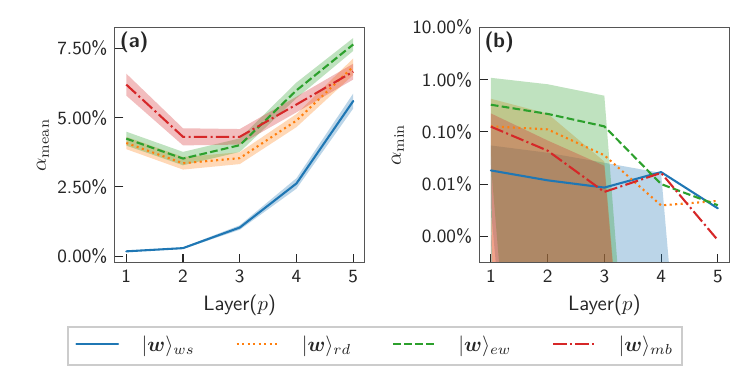}
	\caption[The comparison between different initial states]{A comparison of the impact of initial states on the performance of our QAOA ansatz: (a) $\alpha_{\mean}$, (b) $\alpha_{\min}$. We present the optimised $\alpha_{\mean}$ and corresponding $\alpha_{\min}$ over randomly selected GMVP instances from the market as a function of the number of the layer $p$. We compare the warm-started $|\boldsymbol{w}\rangle_{ws}$, randomly initialised $|\boldsymbol{w}\rangle_{rd}$, the equally weighted $|\boldsymbol{w}\rangle_{ew}$, and the maxbias $|\boldsymbol{w}\rangle_{mb}$ initial states. The shaded region in (a) indicates the standard error of the data, and the shaded region in (b) indicates the 50\% to 90\% percentile of data. 
    We only use warm-started states that are not already equal to the global minimum to facilitate a fair comparison. The QAOA ansatz uses the mixing operator $\hat{M}_1(\beta)$ with the ansatz initial parameters randomly generated. The optimiser is DA, and the optimisation method is the fix{ed-$p$} ansatz approach with the number of shots during optimisation $N_m =2^{4}=16$ and the final expectation value measurement shots $N_M = 2^{16} = 65536$.  
	}\label{fig:initialCom}
\end{figure}

\section{Quantum bridge on two-qubit excitations}\label{app:two-qubitex}
In this appendix, we analyse the action of two-qubit excitations in a specific design pattern.
We define the following pair of operator to exchange qubit excitation between two qubits:
{\begin{equation}\label{app:mixerOp1}
	\begin{aligned}
		 & \hat{S}_{AB}^- \coloneqq i\left(\hat{Q}^{\dagger}_{A}\hat{Q}_{B} - \hat{Q}_{A}\hat{Q}^{\dagger}_{B}\right) \\
		 & \hat{S}_{AB}^+ \coloneqq i \left(\hat{Q}^{\dagger}_{A}\hat{Q}_{B} + \hat{Q}_{A}\hat{Q}^{\dagger}_{B}\right),
	\end{aligned}
\end{equation}}
where $A$ and $B$ are the qubits sub-systems acted upon non-trivially.
By expanding them with Pauli operators, we have
{\begin{align}
	 & \hat{S}^+_{AB} = \frac{i}{2}(\hat{X}_{A}\hat{X}_{B}+ \hat{Y}_{A}\hat{Y}_{B}) \\
	 & \hat{S}^-_{AB} = -\frac{1}{2}(\hat{X}_{A}\hat{Y}_{B}- \hat{Y}_{A}\hat{X}_{B}).
\end{align}}

If we apply a Pauli-$Z$ operator to one of the qubit sub-systems of {$\hat{S}^{\pm}_{AB}$}, we have,
{\begin{equation}\label{SZrelation1}
	\begin{aligned}
		\hat{S}_{AB}^{-} \hat{Z}_{A} & =-\hat{S}_{AB}^{-} \hat{Z}_{B}= -\frac{1}{2} (\hat{X}_A\hat{Y}_B -\hat{Y}_A\hat{X}_B) \hat{Z}_A  = \frac{i}{2}(\hat{X}_A\hat{X}_B +\hat{Y}_A\hat{Y}_B) = \hat{S}^{+}_{AB},     \\
		\hat{S}_{AB}^{+} \hat{Z}_{A} & = \hat{S}_{AB}^{+} \hat{Z}_{B} = \frac{i}{2} (\hat{X}_A\hat{X}_B +\hat{Y}_A\hat{Y}_B) \hat{Z}_A =-\frac{1}{2}(\hat{X}_A\hat{Y}_B -\hat{Y}_A\hat{X}_B) = \hat{S}^{-}_{AB}, \\
		\hat{Z}_{A}  \hat{S}_{AB}^{-} & =-\hat{Z}_{B}  \hat{S}_{AB}^{-} = -  \hat{S}_{AB}^{-} \hat{Z}_{A}  = - \hat{S}^{+}_{AB}, \\
		\hat{Z}_{A}  \hat{S}_{AB}^{+} & =\hat{Z}_{B}  \hat{S}_{AB}^{+}  = - \hat{S}_{AB}^{+} \hat{Z}_{A} =  \hat{S}^{-}_{AB}.
	\end{aligned}
\end{equation}}
If we apply {$\hat{Z}_A\hat{Z}_B$}, then,
{\begin{equation}\label{SZrelation13}
	\begin{aligned}
		\hat{S}_{AB}^{-} \hat{Z}_{A} \hat{Z}_{B} & =\hat{Z}_{A} \hat{Z}_{B} \hat{S}_{AB}^{-} =-\frac{1}{2} (\hat{X}_A\hat{Y}_B -\hat{Y}_A\hat{X}_B) \hat{Z}_A \hat{Z}_B =-\frac{1}{2} (\hat{Y}_A \hat{X}_B -  \hat{X}_A \hat{Y}_B ) = - \hat{S}^{-}_{AB} \\
		\hat{S}_{AB}^{+} \hat{Z}_{A} \hat{Z}_{B} & =\hat{Z}_{A} \hat{Z}_{B} \hat{S}_{AB}^{+} =\frac{i}{2} (\hat{X}_A\hat{X}_B +\hat{Y}_A\hat{Y}_B)\hat{Z}_A \hat{Z}_B = - \hat{S}^{+}_{AB}.
	\end{aligned}
\end{equation}}
We can expand the exponential of {$\hat{S}^{-}_{AB}$} as
{\begin{equation}\label{SwapBridge1}
	\begin{aligned}
		e^{-i \beta  \hat{S}^{-}_{AB}} = & e^{\frac{i \beta}{2}(\hat{X}_{A}\hat{Y}_{B}-\hat{Y}_{A}\hat{X}_{B})}=e^{\frac{i \beta}{2}\hat{X}_{A}\hat{Y}_{B}}e^{-\frac{i \beta}{2}\hat{Y}_{A}\hat{X}_{B}}                                                   \\
		=                         & ( \cos\frac{\beta}{2}\hat{\mathds{1}} + i \sin\frac{\beta}{2}\hat{X}_{A}\hat{Y}_{B} )(  \cos\frac{\beta}{2}\hat{\mathds{1}} - i \sin\frac{\beta}{2}\hat{Y}_{A}\hat{X}_{B} )                        \\
		=                         & \cos^2\frac{\beta}{2}\hat{\mathds{1}} + \sin^2\frac{\beta}{2}\hat{Z}_A \hat{Z}_B +i\sin\frac{\beta}{2}\cos\frac{\beta}{2} \hat{X}_A\hat{Y}_B -i\sin\frac{\beta}{2}\cos\frac{\beta}{2} \hat{Y}_A\hat{X}_B \\
		=                         & \cos^2\frac{\beta}{2}\hat{\mathds{1}} + \sin^2\frac{\beta}{2}\hat{Z}_A \hat{Z}_B +i\frac{\sin\beta}{2} (\hat{X}_A\hat{Y}_B - \hat{Y}_A\hat{X}_B)                                                         \\
		=                         & \cos^2\frac{\beta}{2}\hat{\mathds{1}} + \sin^2\frac{\beta}{2}\hat{Z}_A \hat{Z}_B - i \sin\beta  \hat{S}^{-}_{AB}                        
	\end{aligned}
\end{equation}}
where {$\hat{\mathds{1}}$ and $\hat{Z}_A \hat{Z}_B$} leave the occupations of the computational basis states invariant, {$\hat{S}^{-}_{AB}$} varies the occupations, and the parameter $\beta$ controls the magnitude of the change. 

If we apply another operator {$e^{-i\beta  \hat{S}^{-}_{BC}}$} to the right side of {$e^{-i\beta  \hat{S}^{-}_{AB}}$} [see \Cref{fig:AppQE2}(a)], we have
{\begin{equation}\label{SwapBridge2}
	\begin{aligned}
		  & e^{-i\beta  \hat{S}^{-}_{AB}}e^{-i\beta  \hat{S}^{-}_{BC}}                                           \\
		= & e^{\frac{i \beta}{2}(\hat{X}_{A}\hat{Y}_{B}-\hat{Y}_{A}\hat{X}_{B})}e^{\frac{i \beta}{2}(\hat{X}_{B}\hat{Y}_{C}-\hat{Y}_{B}\hat{X}_{C})} \\
		= & \left(\cbf{2}{2}\hat{\mathds{1}} + \sbf{2}{2}\hat{Z}_A \hat{Z}_B -i\sbb{} \hat{S}^{-}_{AB} \right) \left( \cbf{2}{2}\hat{\mathds{1}} + \sbf{2}{2} \hat{Z}_B \hat{Z}_C -i \sbb{}  \hat{S}^{-}_{BC} \right)               \\
  = & \cbf{4}{2}\hat{\mathds{1}}+\frac{1}{4}\sbb{2}\hat{Z}_B \hat{Z}_C  +\frac{1}{4}\sbb{2} \hat{Z}_A \hat{Z}_B +  \sbf{4}{2}\hat{Z}_A \hat{Z}_C -i \sbb{}\cbf{2}{2}\hat{S}^{-}_{BC} -i \sbb{} \cbf{2}{2} \hat{S}^{-}_{AB} \\
  &\quad   +i \sbf{2}{2}\sbb{} \hat{Z}_A\hat{S}^{+}_{BC} +i \sbf{2}{2}\sbb{} \hat{S}^{+}_{AB}\hat{Z}_C- \sbb{2} \hat{S}^{-}_{AB}\hat{S}^{-}_{BC}.
	\end{aligned}
\end{equation}}

If we continue to add another operator {$e^{-i\beta  \hat{S}^{-}_{BC}}$} to the left of {$e^{-i\beta  \hat{S}^{-}_{AB}}$} [see \Cref{fig:AppQE2}(b)], we have
{\begin{equation}\label{SwapBridge3}
	\begin{aligned}
		  & e^{-i\beta  \hat{S}^{-}_{BC}} e^{-i\beta \hat{S}^{-}_{AB}}e^{-i\beta  \hat{S}^{-}_{BC}}                                                                             \\
        = & \left(e^{-i\beta  \hat{S}^{-}_{BC}} e^{-i\beta  \hat{S}^{-}_{AB}}\right) e^{-i\beta  \hat{S}^{-}_{BC}}  \\
	= & \cbf{6}{2}\hat{\mathds{1}} + \frac{1}{4}\cbf{2}{2}\sbb{2} \hat{Z}_A\hat{Z}_B -i \sbb{}\cbf{4}{2} \hat{S}^{-}_{AB} \\ 
    & \quad + \frac{1}{4} \sbb{2}\cbf{2}{2} \hat{Z}_B\hat{Z}_C + \frac{1}{4} \sbb{2} \sbf{2}{2} \hat{Z}_A\hat{Z}_C - \frac{i}{4} \sbb{3} \hat{Z}_C \hat{S}^{+}_{AB} \\
    &  \quad + \frac{1}{4} \sbb{2} \cbf{2}{2}\hat{Z}_A\hat{Z}_B + \frac{1}{4} \sbb{2}\sbf{2}{2} \hat{\mathds{1}} + \frac{i}{4} \sin^3\beta \hat{S}^{-}_{AB} \\
    & \quad + \sbf{4}{2}\cbf{2}{2}\hat{Z}_A\hat{Z}_C +\sbf{6}{2} \hat{Z}_B\hat{Z}_C +i\sbf{4}{2}\sbb{} \hat{Z}_C \hat{S}^{+}_{AB} \\
    & \quad -i\sbb{} \cbf{4}{2} \hat{S}^{-}_{BC} -i\frac{1}{4}\sbb{3} \hat{S}^{+}_{BC}\hat{Z}_A - \sbb{2}\cbf{2}{2}\hat{S}^{-}_{BC}\hat{S}^{-}_{AB} \\
     & \quad -i \sbb{} \cbf{4}{2} \hat{S}^{-}_{AB} +i \frac{1}{4} \sbb{3}\hat{S}^{-}_{AB} - \sbb{2} \cbf{2}{2}\hat{S}^{-}_{AB}\hat{S}^{-}_{AB} \\
    & \quad +i\frac{1}{4} \sbb{3} \hat{Z}_A\hat{S}^{+}_{BC} -i \sbf{4}{2}\sbb{} \hat{S}^{-}_{BC} -\sbf{2}{2} \sbb{2} \hat{S}^{+}_{BC}\hat{S}^{+}_{AB} \\
    & \quad  +i\frac{1}{4} \sbb{3}\hat{S}^{+}_{AB}\hat{Z}_C -i \sbf{4}{2} \sbb{}\hat{S}^{+}_{AB}\hat{Z}_C + \sbf{2}{2} \sbb{2} \hat{S}^{+}_{AB}\hat{S}^{-}_{AB} \hat{Z}_C\\
    & \quad - \sbb{2} \cbf{2}{2} \hat{S}^{-}_{AB}\hat{S}^{-}_{BC} - \sbb{2} \sbf{2}{2}\hat{S}^{-}_{AB}\hat{S}^{-}_{BC} \hat{Z}_A\hat{Z}_B \\
     = & \left(\cos^6\frac{\beta}{2} + \frac{1}{4} \sin^2\beta \sin^2\frac{\beta}{2} \right) \hat{\mathds{1}} + \frac{1}{2}\cos^2\frac{\beta}{2}\sin^2\beta \hat{Z}_A\hat{Z}_B \\
    & \quad  + \frac{1}{2} \sbb{2} \sbf{2}{2} \hat{Z}_A\hat{Z}_C + \left(\frac{1}{4} \sbb{2} \cbf{2}{2} + \sbf{6}{2}\right) \hat{Z}_B\hat{Z}_C \\
    & \quad -i \left( 2\sbb{}\cbf{4}{2} - \frac{1}{2}\sbb{3} \right)\hat{S}^{-}_{AB} -i \left(  \sbb{} \cbf{4}{2}  + \sbf{4}{2} \sbb{}\right)\hat{S}^{-}_{BC} \\
    & \quad  -\sbb{2} \cbf{2}{2}\hat{S}^{-}_{AB}\hat{S}^{-}_{AB} + \sbf{2}{2} \sbb{2}\hat{S}^{+}_{AB}\hat{S}^{-}_{AB} \hat{Z}_C - \sbb{2}\cbf{2}{2} \hat{S}^{-}_{BC}\hat{S}^{-}_{AB} \\ 
    &- \sbf{2}{2} \sbb{2}\hat{S}^{+}_{BC}\hat{S}^{+}_{AB} - \sbb{2}\cbf{2}{2}  \hat{S}^{-}_{AB}\hat{S}^{-}_{BC} - \sbb{2} \sbf{2}{2}\hat{S}^{-}_{AB}\hat{S}^{-}_{BC} \hat{Z}_A\hat{Z}_B \\
    = & k^\beta_{1} \hat{\mathds{1}} +  k^\beta_{2}\hat{Z}_A\hat{Z}_B + k^\beta_{3}\hat{Z}_B\hat{Z}_C -i k^\beta_{4} \hat{S}_{AB}^{-} -ik^\beta_{5} \hat{S}_{BC}^{-} -ik^\beta_{6} \hat{S}_{AC}^{+},
	\end{aligned}
\end{equation}}
where 
{\begin{equation}
\begin{aligned}
k_1^\beta & = \cos^6\frac{\beta}{2} + \frac{1}{4} \sin^2\beta \sin^2\frac{\beta}{2} -\frac{1}{2} \sbb{2}\cbf{2}{2} \\
k_2^\beta & = \cos^2\frac{\beta}{2}\sin^2\beta  \\
k_3^\beta & = \frac{1}{4} \sbb{2} \cbf{2}{2} + \sbf{6}{2}+\frac{1}{2}\sbf{2}{2} \sbb{2} \\
k^\beta_{4} & = 2\sbb{}\cbf{4}{2} - \frac{1}{2}\sbb{3} \\
k^\beta_{5} & =  \sbb{} \cbf{4}{2}  -\sbf{4}{2} \sbb{} \\
k^\beta_{6} & = \sbb{2}.
\end{aligned}
\end{equation}}
In the above derivation, we have used the following identities in addition to those in \Cref{SZrelation1,SZrelation13}:
{\begin{equation}\label{SSrelation1}
\begin{aligned}
\hat{S}_{AB}^{+}\hat{S}_{AB}^{-} &= - \hat{S}_{AB}^{-}\hat{S}_{AB}^{+} = -\frac{1}{2}\left(\hat{Z}_A-\hat{Z}_B \right), \\
\hat{S}_{AB}^{+}\hat{S}_{AB}^{+} &=-\hat{S}_{AB}^{-}\hat{S}_{AB}^{-} = -\frac{1}{2}\left(\hat{\mathds{1}} -\hat{Z}_A \hat{Z}_B\right)\\
\hat{S}_{BC}^{-}\hat{S}_{AB}^{-} & = \frac{1}{2} \hat{S}_{AC}^{+} -\frac{1}{2} \hat{S}_{AC}^{-} \hat{Z}_B, \quad \hat{S}_{BC}^{+}\hat{S}_{AB}^{+} = \frac{1}{2} \hat{S}_{AC}^{+} -\frac{1}{2}\hat{S}_{AC}^{-}\hat{Z}_B \\
\hat{S}^{-}_{AB}\hat{S}^{-}_{BC} & =\hat{S}^{-}_{AB}\hat{S}^{-}_{BC} \hat{Z}_A\hat{Z}_B = \frac{1}{2} \hat{S}_{AC}^{-}\hat{Z}_B +\frac{1}{2} \hat{S}_{AC}^{+}\\
\hat{S}^{-}_{AB}\hat{S}^{-}_{BC}\hat{S}^{-}_{AB} & =  0.
\end{aligned}
\end{equation}}

In~\Cref{SwapBridge3}, the support {on $\hat{\mathds{1}}, \hat{Z}_A \hat{Z}_B,$ and $ \hat{Z}_B\hat{Z}_C$} leave the occupations of the computational basis states invariant, the support on {$\hat{S}^{-}_{AB}, \hat{S}^{-}_{BC}, $ and $\hat{S}^{+}_{AC}$} varies the occupations, and the parameter $\beta$ controls the magnitude of the change of total superposition. The decomposed supports are linearly independent of each other.
In \Cref{SwapBridge3}, we observe in addition to the expected support on {$\hat{S}_{AB}^{-}$ and $\hat{S}_{BC}^{-}$} as shown directly in the circuit~\Cref{fig:AppQE2}(b), the two {$e^{-i\beta  \hat{S}^{-}_{BC}}$ operators in $e^{-i\beta  \hat{S}^{-}_{BC}} e^{-i\beta  \hat{S}^{-}_{AB}}e^{-i\beta  \hat{S}^{-}_{BC}}$} act as a bridge to generate support on {$\hat{S}_{AC}^{+}$}, and we name this phenomenon as the quantum bridge on two-qubit excitations.
In fact, this result, generating new support, can also be extended to additional ancilla by repeatedly sandwiching the circuit with additional two-qubit excitation operators.

\begin{figure}[h]
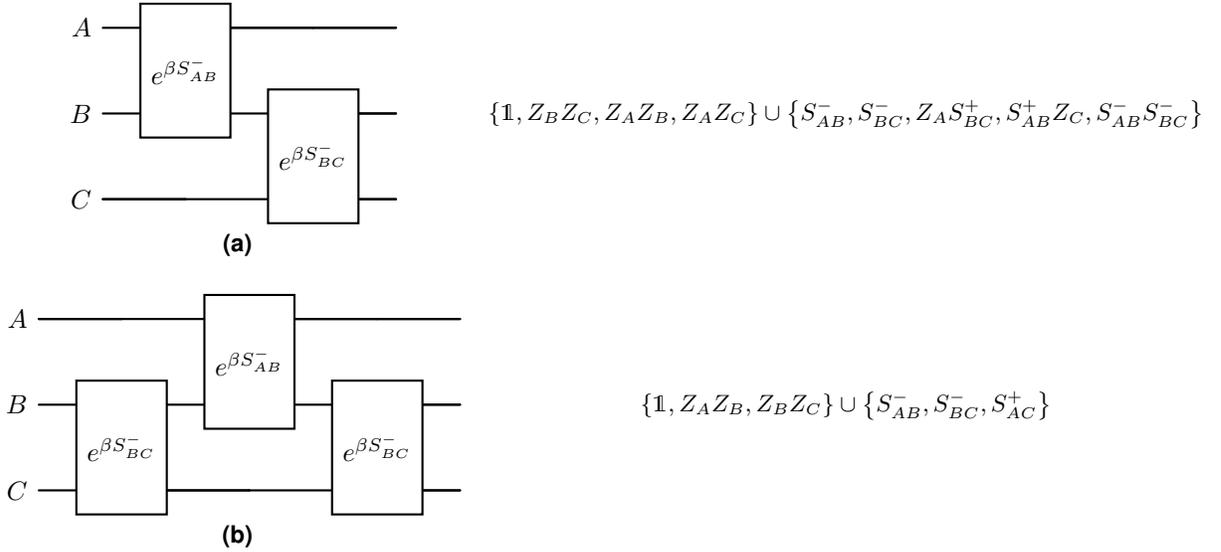

	\centering
	\begin{tabular}{c c}
		\usebox\boxA & {$\left\{\hat{\mathds{1}},\hat{Z}_B \hat{Z}_C,\hat{Z}_A \hat{Z}_B,\hat{Z}_A \hat{Z}_C \right\}\cup \left\{ \hat{S}^{-}_{AB},\hat{S}^{-}_{BC}, \hat{Z}_A\hat{S}^{+}_{BC},\hat{S}^{+}_{AB}\hat{Z}_C,\hat{S}^{-}_{AB}\hat{S}^{-}_{BC} \right\} $} \\
		\textbf{\fontfamily{phv}\selectfont (a)}          &           \\
		\addlinespace
		\addlinespace
		\usebox\boxC & {$ \left\{\hat{\mathds{1}},\hat{Z}_A\hat{Z}_B, \hat{Z}_B\hat{Z}_C \right\}\cup \left\{\hat{S}^{-}_{AB}, \hat{S}^{-}_{BC}, \hat{S}_{AC}^{+} \right\} $}  \\
		\textbf{\fontfamily{phv}\selectfont (b)}          & 
	\end{tabular}
	\caption[two-qubit excitation]{Quantum circuit diagrams of (a) {$e^{-i\beta  \hat{S}^{-}_{AB}}e^{-i\beta  \hat{S}^{-}_{BC}}$} and (b) {$e^{-i\beta  \hat{S}^{-}_{BC}} e^{-i\beta  \hat{S}^{-}_{AB}}e^{-i\beta  \hat{S}^{-}_{BC}}$}. The set on the right of each circuit is the linear independent support of the unitary generated by the quantum circuit. These operators form an orthogonal basis for the subspace with respect to the Hilbert-Schmidt inner product. The first set includes operators that leave the occupation of computational basis states invariant. The second set includes operators that vary the occupation.} \label{fig:AppQE2}
\end{figure}

\section{Quantum bridge on three-qubit excitations}\label{app:three-qubitex}
In this appendix, we analyse the action of three-qubit excitations. Similar to the two-qubit excitation, for the three-qubit excitation operator as in \Cref{mixereq1}, we can define a pair of operators:
{\begin{equation}\label{app:mixerOp2}
	\begin{aligned}
		 & \hat{P}_{ABC}^{\pm} \coloneqq i\left(\hat{Q}^{^\dagger}_{A} \hat{Q}_{B} \hat{Q}_{C}\pm \hat{Q}_{A} \hat{Q}^{\dagger}_{B} \hat{Q}^{\dagger}_{C}\right). 
	\end{aligned}
\end{equation}}
In the Pauli basis, these operators are
{\begin{align}
	 & \hat{P}^+_{ABC} = -\frac{i}{4}( \hat{X}_{A} \hat{X}_{B} \hat{X}_{C}+ \hat{Y}_{A} \hat{X}_{B} \hat{Y}_{C}- \hat{X}_{A} \hat{Y}_{B} \hat{Y}_{C}+\hat{Y}_{A} \hat{Y}_{B} \hat{X}_{C}), \\
	 & \hat{P}^-_{ABC}  = \frac{1}{4} (\hat{X}_{A}\hat{X}_{B}\hat{Y}_{C}+\hat{X}_{A}\hat{Y}_{B}\hat{X}_{C}- \hat{Y}_{A}\hat{X}_{B}\hat{X}_{C}+\hat{Y}_{A}\hat{Y}_{B}\hat{Y}_{C}).
\end{align}}
Notice {$\hat{P}_{ABC}^{\pm}$} is invariant under swapping the $B$ and $C$ indices. Applying Pauli-$Z$ operators we find the following identities:
{\begin{equation}\label{SwapRela6}
	\begin{aligned}
		 & \hat{P}^{+}_{ABC} \hat{Z}_A =  -\hat{Z}_A  \hat{P}^{+}_{ABC} =  \hat{P}^{-}_{ABC}              \\
		 & \hat{P}^{+}_{ABC} \hat{Z}_B = \hat{P}^{+}_{ABC} \hat{Z}_C =  -\hat{Z}_B  \hat{P}^{+}_{ABC} = -\hat{Z}_C  \hat{P}^{+}_{ABC} = - \hat{P}^{-}_{ABC}             \\
		 & \hat{P}^{+}_{ABC} \hat{Z}_A\hat{Z}_B = \hat{P}^{+}_{ABC} \hat{Z}_A\hat{Z}_C=  \hat{Z}_A\hat{Z}_B  \hat{P}^{+}_{ABC} = \hat{Z}_A\hat{Z}_C  \hat{P}^{+}_{ABC}= - \hat{P}^{+}_{ABC}        \\
		 & \hat{P}^{+}_{ABC} \hat{Z}_B\hat{Z}_C =  \hat{Z}_B\hat{Z}_C  \hat{P}^{+}_{ABC} =  \hat{P}^{+}_{ABC}         \\
		 & \hat{P}^{+}_{ABC} \hat{Z}_A\hat{Z}_B\hat{Z}_C  =  \hat{Z}_A\hat{Z}_B\hat{Z}_C  \hat{P}^{+}_{ABC} =  \hat{P}^{-}_{ABC}. 
	\end{aligned}
\end{equation}}
Now we will use these identities to express the circuit in \Cref{fig:AppQE3}(a):
{\begin{equation}\label{SwapRela7}
	\begin{aligned}
		  & e^{-i\beta  \hat{S}^-_{CD}}e^{-i\beta  \hat{P}^-_{A B C}}e^{-i\beta\hat{S}^-_{CD}}                                                                \\
    = &\left(\cbf{2}{2}\hat{\mathds{1}} + \sbf{2}{2}\hat{Z}_C \hat{Z}_D -i  \sbb{}  \hat{S}^{-}_{CD} \right) \\
    & \quad  \left( \frac{1}{4}\left(3+\cbb{}\right)\hat{\mathds{1}}_{A B C} +\frac{1}{2}\sbf{2}{2}\left(\hat{Z}_{A}\hat{Z}_{B} +\hat{Z}_{A}\hat{Z}_{C}- \hat{Z}_{B}\hat{Z}_{C}\right) - i\sbb{}  \hat{P}^{-}_{A B C}\right) \\
    & \quad 
 \left(\cbf{2}{2}\hat{\mathds{1}} + \sbf{2}{2}\hat{Z}_C \hat{Z}_D -  i\sbb{}  \hat{S}^{-}_{CD}  \right) \\
 = &  \left(\frac{1}{4} \cbf{2}{2} \left(3+\cbb{}\right)\hat{\mathds{1}} +\frac{1}{4} \left(3+\cbb{}\right)\sbf{2}{2} \hat{Z}_C\hat{Z}_D -i\frac{1}{4} (3+\cbb{})\sbb{} \hat{S}_{CD}^{-} \right.\\
 & \quad  +\frac{1}{8} \sbb{2}(\hat{Z}_A\hat{Z}_B+\hat{Z}_A\hat{Z}_C-\hat{Z}_B\hat{Z}_C) + \frac{1}{2} \sbf{4}{2}(\hat{Z}_A\hat{Z}_B\hat{Z}_C\hat{Z}_D +\hat{Z}_A\hat{Z}_D-\hat{Z}_B\hat{Z}_D) \\ 
 & \quad  -i\frac{1}{2} \sbf{2}{2} \sbb{}(\hat{S}_{CD}^{-} \hat{Z}_A\hat{Z}_B + \hat{S}_{CD}^{+}\hat{Z}_A -\hat{S}_{CD}^{+}\hat{Z}_B) \\
 & \quad  -i\left.\cbf{2}{2}\sbb{}\hat{P}_{ABC}^{-}-i\sbf{2}{2}\sbb{}\hat{P}_{ABC}^{+}\hat{Z}_D -\sbb{2} \hat{S}_{CD}^{-}\hat{P}_{ABC}^{-}\right) \\
 & \quad 
 \left(\cbf{2}{2}\hat{\mathds{1}} + \sbf{2}{2}\hat{Z}_C \hat{Z}_D -i  \sbb{}  \hat{S}^{-}_{CD}  \right) \\
 & =m_1^\beta \hat{\mathds{1}} +m_2^\beta \hat{Z}_A\hat{Z}_B+m_3^\beta \hat{Z}_C\hat{Z}_D+m_4^\beta \hat{Z}_A\hat{Z}_C+m_5^\beta \hat{Z}_B\hat{Z}_C+m_6^\beta \hat{Z}_A\hat{Z}_B\hat{Z}_C\hat{Z}_D\\
  & \quad  -m_7^\beta i\hat{S}_{CD}^{-} - m_8^\beta i \hat{Z}_A\hat{Z}_B\hat{S}_{CD}^{-} -  m_{9}^\beta i \hat{P}_{ABC}^{-} - m_{10}^\beta i \hat{P}_{ABD}^{+},
	\end{aligned}
\end{equation}}
where 
{\begin{equation}
\begin{aligned}
m_1^\beta & = \frac{1}{4}(3+\cbb{})(\frac{1}{4}(3+\cos({2\beta}) -\frac{1}{2} \sbb{2}), \\
m_2^\beta &= \frac{1}{2} \sbf{6}{2} +\frac{1}{8}\sbb{2}\cbf{2}{2} -\frac{1}{4}\sbf{2}{2}\sbb{2}, \\
m_3^\beta & = \frac{1}{4} \sbb{2}(3+\cbb{}), \\
m_4^\beta & = \frac{1}{8} \sbb{2}\cbf{2}{2}+\frac{1}{2}\sbf{6}{2}+\frac{1}{4}\sbf{2}{2}\sbb{2}, \\
m_5^\beta & = -\frac{1}{8} \sbb{2}\cbf{2}{2}-\frac{1}{2}\sbf{6}{2}-\frac{1}{4}\sbf{2}{2}\sbb{2}, \\
m_6^\beta & = \frac{1}{2} \sbf{2}{2}\sbb{2}, m_7^\beta = \frac{1}{4}(3+\cbb{})\sin(2\beta)  \\
m_8^\beta & = \frac{1}{4}\sbb{3} -\sbf{4}{2}\sbb{}, m_{9}^\beta  =\cbb{}\sbb{}, m_{10}^\beta = \sbb{2}.
\end{aligned}
\end{equation}}
In the above derivation, we have used the following identities in addition to those in \Cref{SwapRela6}:
{\begin{equation}\label{SPrelation1}
\begin{aligned}
\hat{S}_{CD}^{-}\hat{P}_{ABC}^{-} & = \frac{1}{2}(\hat{P}_{ABD}^{+} -\hat{Z}_C\hat{P}_{ABD}^{-} ) \\
\hat{P}_{ABC}^{-}\hat{S}_{CD}^{-} & = \frac{1}{2} (\hat{P}_{ABD}^{+} +\hat{Z}_C\hat{P}_{ABD}^{-}) \\
\hat{P}_{ABC}^{+}\hat{S}_{CD}^{+} & =(-\hat{P}_{ABC}^{-} \hat{Z}_C)(-\hat{Z}_C \hat{S}_{CD}^{-} ) = \hat{P}_{ABC}^{-}\hat{S}_{CD}^{-} = \frac{1}{2} (\hat{P}_{ABD}^{+} +\hat{Z}_C\hat{P}_{ABD}^{-}) \\
\hat{S}_{CD}^{-}\hat{P}_{ABC}^{-}\hat{S}_{CD}^{-} & = \frac{1}{2}(\hat{P}_{ABD}^{+} -\hat{Z}_C\hat{P}_{ABD}^{-} ) \hat{S}_{CD}^{-} \\
& = \frac{1}{2}( \hat{P}_{ABD}^{+}\hat{S}_{CD}^{-} - \hat{P}_{ABD}^{-}\hat{Z}_C \hat{S}_{CD}^{-})\\
& = \frac{1}{2}( -\hat{P}_{ABD}^{-} \hat{Z}_D\hat{S}_{CD}^{-} - \hat{P}_{ABD}^{-}\hat{Z}_C \hat{S}_{CD}^{-}) \\
&= \frac{1}{2}( -\hat{P}_{ABD}^{-} \hat{S}_{CD}^{+} - \hat{P}_{ABD}^{-}(-\hat{S}_{CD}^{+})) = 0.
\end{aligned}
\end{equation}}

In~\Cref{SwapRela7}, the support {on $\hat{\mathds{1}}, \hat{Z}_A \hat{Z}_B, \hat{Z}_C\hat{Z}_D,\hat{Z}_A\hat{Z}_C,\hat{Z}_B\hat{Z}_C,$ and $\hat{Z}_A\hat{Z}_B\hat{Z}_C\hat{Z}_D$} leave the occupations of the computational basis states invariant, the support {on $\hat{S}^{-}_{CD}, \hat{Z}_A\hat{Z}_B\hat{S}^{-}_{CD}, \hat{P}_{ABC}^{-},$ and $\hat{P}_{ABD}^{+}$} varies the occupations, and the parameter $\beta$ controls the magnitude of the change of the total superposition. The decomposed supports are linearly independent of each other.
Thus, {$e^{-i\beta  \hat{S}^-_{CD}}e^{-i\beta  \hat{P}^-_{A B C}}e^{-i\beta\hat{S}^-_{CD}}$} has support on {$\hat{P}_{ABD}^{+}$} in addition to the expected support on {$\hat{P}_{ABC}^{-}$ and $\hat{S}_{CD}^-$} and local phase generating terms. The same derivation follows directly for {$e^{-i\beta  \hat{S}^-_{BD}}e^{-i\beta  \hat{P}^-_{A B C}}e^{-i\beta\hat{S}^-_{BD}}$} [see \Cref{fig:AppQE3}(b)] as {$\hat{P}^-_{A B C}$} is invariant under exchange of the indices $B$ and $C$.

Next we consider the circuit in \Cref{fig:AppQE3}(c):
{\begin{equation}\label{SwapRela8}
	\begin{aligned}
		  & e^{-i\beta  \hat{S}^-_{AD}}e^{-i\beta  \hat{P}^-_{A B C}}e^{-i\beta\hat{S}^-_{AD}}                                                                \\
    = &\left(\cbf{2}{2}\hat{\mathds{1}} + \sbf{2}{2}\hat{Z}_A \hat{Z}_D -  i\sbb{}  \hat{S}^{-}_{AD} \right) \\
    & \quad \left( \frac{1}{4}\left(3+\cbb{}\right)\hat{\mathds{1}}_{A B C} +\frac{1}{2}\sbf{2}{2}\left(\hat{Z}_{A}\hat{Z}_{B} +\hat{Z}_{A}\hat{Z}_{C}- \hat{Z}_{B}\hat{Z}_{C}\right)-i \sbb{}  \hat{P}^{-}_{A B C}\right) \\
    & \quad  
 \left(\cbf{2}{2}\hat{\mathds{1}} + \sbf{2}{2}\hat{Z}_A \hat{Z}_D -i  \sbb{}  \hat{S}^{-}_{AD}  \right) \\
 = &  \left(\frac{1}{4} \cbf{2}{2} (3+\cbb{})\hat{\mathds{1}} +\frac{1}{4} (3+\cbb{})\sbf{2}{2} \hat{Z}_A\hat{Z}_D +\frac{1}{4} (3+\cbb{})\sbb{} \hat{S}_{AD}^{-} \right.\\
 & \quad +\frac{1}{8} \sbb{2}(\hat{Z}_A\hat{Z}_B+\hat{Z}_A\hat{Z}_C-\hat{Z}_B\hat{Z}_C) + \frac{1}{2} \sbf{4}{2}( \hat{Z}_B\hat{Z}_D-\hat{Z}_C\hat{Z}_D-\hat{Z}_A\hat{Z}_B\hat{Z}_C\hat{Z}_D) \\ 
 & \quad -i\frac{1}{2} \sbf{2}{2} \sbb{}( \hat{S}_{AD}^{+}\hat{Z}_B +\hat{S}_{AD}^{+}\hat{Z}_C - \hat{S}_{AD}^{-} \hat{Z}_B\hat{Z}_C) \\
 & \quad  -i\left.\cbf{2}{2}\sbb{}\hat{P}_{ABC}^{-} + i\sbf{2}{2}\sbb{}\hat{P}_{ABC}^{+}\hat{Z}_D -\sbb{2} \hat{S}_{AD}^{-}\hat{P}_{ABC}^{-}\right) \\
 & \quad
 \left(\cbf{2}{2}\hat{\mathds{1}} + \sbf{2}{2}\hat{Z}_A \hat{Z}_D -i  \sbb{}  \hat{S}^{-}_{AD}  \right) \\
 & =m_1^\beta \hat{\mathds{1}} +m_2^\beta \hat{Z}_B\hat{Z}_C+m_3^\beta \hat{Z}_A\hat{Z}_D+m_4^\beta \hat{Z}_A\hat{Z}_B+m_5^\beta \hat{Z}_A\hat{Z}_C+m_6^\beta \hat{Z}_A\hat{Z}_B\hat{Z}_C\hat{Z}_D\\
  & \quad  -m_7^\beta i \hat{S}_{AD}^{-} - m_8^\beta i \hat{Z}_B\hat{Z}_C\hat{S}_{AD}^{-} -  m_{9}^\beta i\hat{P}_{ABC}^{-} - m_{10}^\beta i \hat{P}_{DBC}^{+},
	\end{aligned}
\end{equation}}
where 
{\begin{equation}
\begin{aligned}
m_1^\beta & = \frac{1}{4}(3+\cbb{})(\frac{1}{4}(3+\cos({2\beta}) -\frac{1}{2} \sbb{2}), \\
m_2^\beta &= -\frac{1}{2} \sbf{6}{2} -\frac{1}{8}\sbb{2}\cbf{2}{2} +\frac{1}{4}\sbf{2}{2}\sbb{2}, \\
m_3^\beta & = \frac{1}{4} \sbb{2}(3+\cbb{}), \\
m_4^\beta & = \frac{1}{8} \sbb{2}\cbf{2}{2}+\frac{1}{2}\sbf{6}{2}+\frac{1}{4}\sbf{2}{2}\sbb{2}, \\
m_5^\beta & = \frac{1}{8} \sbb{2}\cbf{2}{2}+\frac{1}{2}\sbf{6}{2}+\frac{1}{4}\sbf{2}{2}\sbb{2}, \\
m_6^\beta & = -\frac{1}{2} \sbf{2}{2}\sbb{2}, m_7^\beta = \frac{1}{4}(3+\cbb{})\sin(2\beta)  \\
m_8^\beta & = -\frac{1}{4}\sbb{3} +\sbf{4}{2}\sbb{}, m_{9}^\beta  =\cbb{}\sbb{}, m_{10}^\beta = -\sbb{2}.
\end{aligned}
\end{equation}}
In the above derivation, we have used the following identities in addition to those in \Cref{SwapRela6} and \Cref{SPrelation1}:
{\begin{equation}\label{SPrelation2}
	\begin{aligned}
		\hat{P}_{ABC}^{-}\hat{S}_{AD}^{-} & = \frac{1}{2} (\hat{Z}_A \hat{P}_{DBC}^{-} - \hat{P}_{DBC}^{+}) \\
		\hat{S}_{AD}^{-}\hat{P}_{ABC}^{-} & = -\frac{1}{2} (\hat{Z}_A \hat{P}_{DBC}^{-} + \hat{P}_{DBC}^{+}) \\
		\hat{P}_{ABC}^{+}\hat{S}_{AD}^{+} & = (\hat{P}_{ABC}^{-} \hat{Z}_A)(-\hat{Z}_A \hat{S}_{AD}^{-}) = -\hat{P}_{ABC}^{-}\hat{S}_{AD}^{-} \\
		\hat{S}_{AD}^{+}\hat{P}_{ABC}^{+} & = (-\hat{S}_{AD}^{-}\hat{Z}_A )(-\hat{Z}_A\hat{P}_{ABC}^{-} ) = \hat{S}_{AD}^{-}\hat{P}_{ABC}^{-} \\
		\hat{S}_{AD}^{-}\hat{P}_{ABC}^{-}\hat{S}_{AD}^{-} & = 0.
		\end{aligned}
\end{equation}}
In~\Cref{SwapRela8}, the support {on $\hat{\mathds{1}},\hat{Z}_B\hat{Z}_C, \hat{Z}_A\hat{Z}_D,\hat{Z}_A\hat{Z}_B,$ and $\hat{Z}_A\hat{Z}_C,\hat{Z}_A\hat{Z}_B\hat{Z}_C\hat{Z}_D$} leave the occupations of the computational basis states invariant, the support {on $\hat{S}^{-}_{AD}, \hat{Z}_B\hat{Z}_C\hat{S}^{-}_{AD}, \hat{P}_{ABC}^{-} ,$ and $\hat{P}_{DBC}^{+}$} vary the occupations, and the parameter $\beta$ controls the magnitude of the change the total superposition. The decomposed supports are linearly independent of each other.
Thus, {$e^{-i\beta  \hat{S}^-_{AD}}e^{-i\beta  \hat{P}^-_{A B C}}e^{-i\beta\hat{S}^-_{AD}}$} has support on {$\hat{P}_{DBC}^{+}$} in addition to the expected support on {$\hat{P}_{ABC}^{-}$ and $\hat{S}_{AD}^-$} and local phase generating terms.

In this appendix, we have shown that by sandwiching a three-qubit excitation with two-qubit excitations that act non-trivially on an ancilla qubit, the resulting unitary will have support on a three-qubit excitation acting upon this ancilla qubit. We name this phenomenon as the quantum bridge on three-qubit excitations. 

\begin{figure}[h]
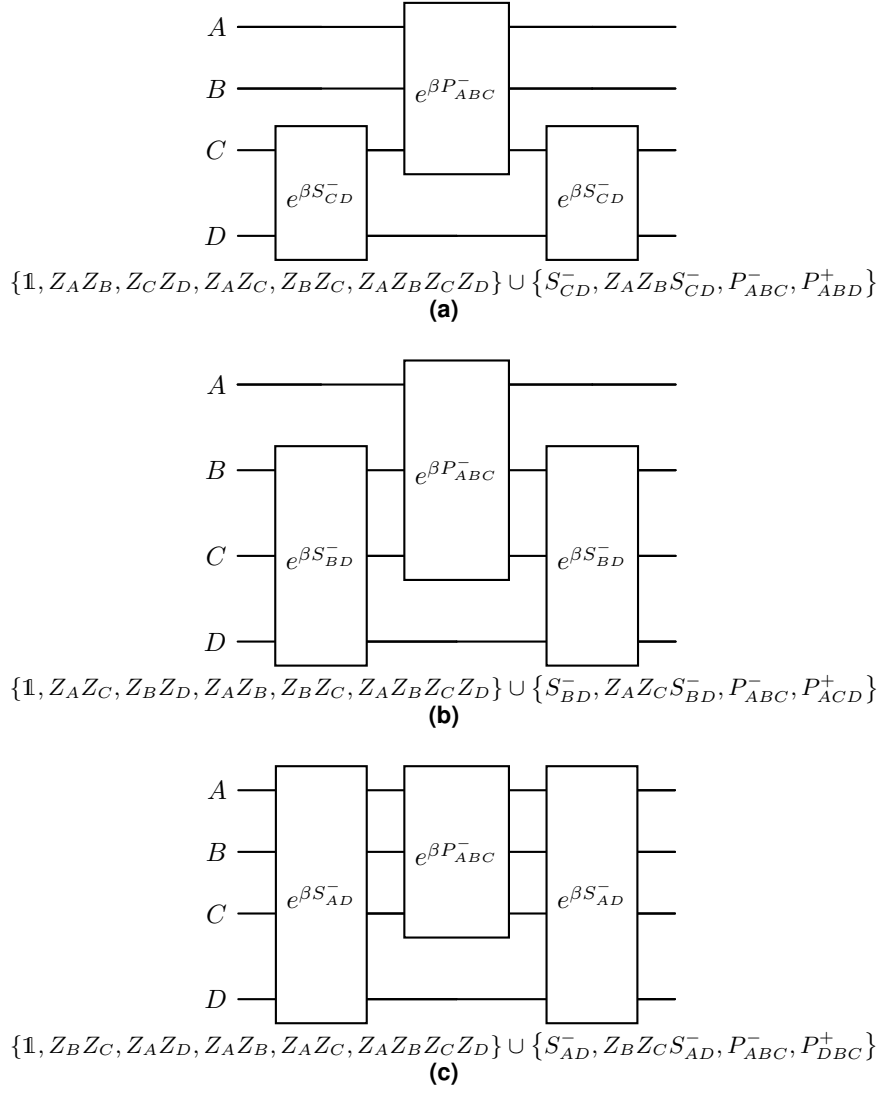

	\centering
	{\begin{tabular}{c c}
		\usebox\boxTa &  \\
		$ \left\{\hat{\mathds{1}},\hat{Z}_A\hat{Z}_B, \hat{Z}_C\hat{Z}_D,\hat{Z}_A\hat{Z}_C,\hat{Z}_B\hat{Z}_C,\hat{Z}_A\hat{Z}_B\hat{Z}_C\hat{Z}_D \right\}\cup \left\{\hat{S}^{-}_{CD}, \hat{Z}_A\hat{Z}_B\hat{S}^{-}_{CD}, \hat{P}_{ABC}^{-}, \hat{P}_{ABD}^{+} \right\}$ & \\
		\textbf{\fontfamily{phv}\selectfont (a)}           &           \\
		\addlinespace
		\addlinespace
		\usebox\boxTc & \\
		$ \left\{\hat{\mathds{1}},\hat{Z}_A\hat{Z}_C, \hat{Z}_B\hat{Z}_D,\hat{Z}_A\hat{Z}_B,\hat{Z}_B\hat{Z}_C,\hat{Z}_A\hat{Z}_B\hat{Z}_C\hat{Z}_D \right\}\cup \left\{\hat{S}^{-}_{BD}, \hat{Z}_A\hat{Z}_C\hat{S}^{-}_{BD}, \hat{P}_{ABC}^{-}, \hat{P}_{ACD}^{+} \right\}$ & \\
		\textbf{\fontfamily{phv}\selectfont (b)}           &           \\
		\addlinespace
		\addlinespace
		\usebox\boxTe & \\
		$ \left\{\hat{\mathds{1}},\hat{Z}_B\hat{Z}_C, \hat{Z}_A\hat{Z}_D,\hat{Z}_A\hat{Z}_B,\hat{Z}_A\hat{Z}_C,\hat{Z}_A\hat{Z}_B\hat{Z}_C\hat{Z}_D \right\}\cup\left\{\hat{S}^{-}_{AD}, \hat{Z}_B\hat{Z}_C\hat{S}^{-}_{AD}, \hat{P}_{ABC}^{-} ,\hat{P}_{DBC}^{+} \right\} $ & \\
		\textbf{\fontfamily{phv}\selectfont (c)}           & 
	\end{tabular}}
	\caption[Three-qubit excitation bridge]{
	Quantum circuits of (a) {$e^{-i\beta  \hat{S}^-_{CD}}e^{-i\beta  \hat{P}^-_{A B C}}e^{-i\beta\hat{S}^-_{CD}}$}, (b) {$e^{-i\beta  \hat{S}^-_{BD}}e^{-i\beta  \hat{P}^-_{A B C}}e^{-i\beta\hat{S}^-_{BD}}$}, and (c) {$e^{-i\beta  \hat{S}^-_{AD}}e^{-i\beta  \hat{P}^-_{A B C}}e^{-i\beta\hat{S}^-_{AD}}$}. The set on the right of each circuit is the linear independent support of the unitary generated by the quantum circuit. These operators form an orthogonal basis for the subspace with respect to the Hilbert-Schmidt inner product. The first set includes operators that leave the occupation of computational basis states invariant. The second set includes operators that vary the occupation. 
	} \label{fig:AppQE3}
\end{figure}

\section{Mixing operator with a implement of quantum bridges}\label{app:qbm}
In this appendix, we show how the design of the mixing operator {$\hat{H}_{tt'}(\beta)$} (see in~\Cref{sec:mix}) utilises the quantum bridges to reduce the cost of quantum resources and generate sufficient support for binary arithmetic operations between asset blocks.
By using the same notation as in \Cref{sec:mix}, we have
{\begin{equation}\label{eq:swprela1}
	\begin{aligned}
		  & \left[\hat{S}^k_{tt'},  \hat{P}^{k}_{ttt'}\right]                                                                                                                                                                                                                                  \\
		= & \left[  \frac{1}{2}\left(\hat{X}^{k}_{t}\hat{Y}^k_{t'}- \hat{Y}^k_{t}\hat{X}^k_{t'}\right), \frac{1}{4} \left(\hat{X}^{k+1}_{t}\hat{X}^k_{t}\hat{Y}^k_{t'}+\hat{X}^{k+1}_{t}\hat{Y}^k_{t}\hat{X}^k_{t'}- \hat{Y}^{k+1}_{t}\hat{X}^k_{t}\hat{X}^k_{t'}+\hat{Y}^{k+1}_{t}\hat{Y}^k_{t}\hat{Y}^k_{t'}\right) \right]                                                      \\
		= & \frac{1}{8}\left\{ \left[\hat{X}^k_{t}\hat{Y}^k_{t'}, \hat{X}^{k+1}_{t}\hat{X}^k_{t}\hat{Y}^k_{t'}\right] +\left[\hat{X}^k_{t}\hat{Y}^k_{t'}, \hat{X}^{k+1}_{t}\hat{Y}^k_{t}\hat{X}^k_{t'}\right] - \left[\hat{X}^k_{t}\hat{Y}^k_{t'}, \hat{Y}^{k+1}_{t}\hat{X}^k_{t}\hat{X}^k_{t'}\right] +\left[\hat{X}^k_{t}\hat{Y}^k_{t'}, \hat{Y}^{k+1}_{t}\hat{Y}^k_{t}\hat{Y}^k_{t'}\right] \right.    \\
		  & \left. \quad \quad -\left[\hat{Y}^k_{t}\hat{X}^k_{t'}, \hat{X}^{k+1}_{t}\hat{X}^k_{t}\hat{Y}^k_{t'}\right] - \left[\hat{Y}^k_{t}\hat{X}^k_{t'}, \hat{X}^{k+1}_{t}\hat{Y}^k_{t}\hat{X}^k_{t'}\right] + \left[\hat{Y}^k_{t}\hat{X}^k_{t'}, \hat{Y}^{k+1}_{t}\hat{X}^k_{t}\hat{X}^k_{t'}\right]-\left[\hat{Y}^k_{t}\hat{X}^k_{t'}, \hat{Y}^{k+1}_{t}\hat{Y}^k_{t}\hat{Y}^k_{t'}\right]   \right\} \\
		= & \frac{1}{8}\left\{ \hat{X}^{k+1}_{t}+  \hat{X}^{k+1}_{t}\left[\hat{X}^k_{t}\hat{Y}^k_{t'}, \hat{Y}^k_{t}\hat{X}^k_{t'}\right] - \hat{Y}^{k+1}_{t} \left[\hat{Y}^k_{t'}, \hat{X}^k_{t'}\right] +\hat{Y}^{k+1}_{t} \left[\hat{X}^k_{t}, \hat{Y}^k_{t}\right] \right.                                                                            \\
		  & \left. \quad \quad -\hat{X}^{k+1}_{t} \left[\hat{Y}^k_{t}\hat{X}^k_{t'}, \hat{X}^k_{t}\hat{Y}^k_{t'}\right] - \hat{X}^{k+1}_{t} + \hat{Y}^{k+1}_{t} \left[\hat{Y}^k_{t}, \hat{X}^k_{t}\right]-\hat{Y}^{k+1}_{t}\left[\hat{X}^k_{t'}, \hat{Y}^k_{t'}\right] \right\}                                                                            \\
		= & 0
	\end{aligned}
\end{equation}}
where the second equality follows from {$[\hat{A}+\hat{B}, \hat{C}+\hat{D}] =[\hat{A},\hat{C}]+[\hat{B},\hat{C}]+[\hat{A},\hat{D}]+[\hat{B},\hat{D}]$ and $[c_1 \hat{A},c_2\hat{B}] = c_1c_2[\hat{A},\hat{B}]$} with $c_1,c_2 \in \mathbb{C}$, for the third equality we use the fact that Pauli operators are involutions, and the last equality follows from {$[\hat{X}\hat{Y}, \hat{Y}\hat{X}] = [\hat{Y}\hat{X},\hat{X}\hat{Y}]$}. Thus, we have
{\begin{equation}\label{eq:swpexprela}
	\begin{aligned}
		  & e^{-i\beta \hat{S}^k_{tt'}} e^{-i\beta \hat{P}^{k}_{ttt'}} \\
		= & e^{-i\beta \hat{S}^k_{tt'}-i\beta \hat{P}^{k}_{ttt'}}      \\
		= & e^{-i\beta \hat{P}^{k}_{ttt'}}e^{-i\beta \hat{S}^k_{tt'}}
	\end{aligned}
\end{equation}}

Then, \Cref{eq:QE23} can be expanded as
{\begin{equation}\label{eq:QE23expan}
	\begin{aligned}
		& \quad \ \hat{H}_{tt'}(\beta)\\ & = \tilde{S}_{tt'}(\beta) \tilde{P}^e_{tt'}(\beta) \tilde{P}^o_{tt'}(\beta) \tilde{S}_{tt'}(\beta)                                        \\
		               & = \left(\bigotimes_{k=1}^l e^{-i\beta \hat{S}^k_{tt'}}  \right)\left(\bigotimes_{k\in \mathcal{K}_1} e^{-i\beta \hat{P}^{k}_{ttt'}} \right) \left( \bigotimes_{k\in \mathcal{K}_2} e^{-i\beta \hat{P}^{k}_{ttt'}} \right) \left(\bigotimes_{k=1}^l e^{-i\beta \hat{S}^k_{tt'}} \right)                                                     \\
		               & = \left( \bigotimes_{k\in \mathcal{K}_2} e^{-i\beta \hat{S}^{k}_{tt'}}  \right) \left( \bigotimes_{k\in \mathcal{K}_1} e^{-i\beta \hat{P}^{k}_{ttt'}} \right) \left( \bigotimes_{k=1}^l e^{-i\beta \hat{S}^k_{tt'}} \right) \left( \bigotimes_{k\in \mathcal{K}_2} e^{-i\beta \hat{P}^{k}_{ttt'}}\right)  \left(\bigotimes_{k\in \mathcal{K}_1} e^{-i\beta \hat{S}^{k}_{tt'}}  \right) \\
		               & = \left( \bigotimes_{k\in \mathcal{K}_2} e^{-i\beta \hat{S}^{k}_{tt'}}  \right)  \left(\bigotimes_{k\in \mathcal{K}_1} e^{-i\beta \hat{P}^{k}_{ttt'}} \right) \left( \bigotimes_{k\in \mathcal{K}_2} e^{-i\beta \hat{S}^{k}_{tt'}}\right) \left(\bigotimes_{k\in \mathcal{K}_1} e^{-i\beta \hat{S}^{k}_{tt'}} \right)                                                   \\ 
					   & \quad \quad \quad \quad \quad  \left( \bigotimes_{k\in \mathcal{K}_2} e^{ -i\beta \hat{P}^{k}_{ttt'}}\right) \left( \bigotimes_{k\in \mathcal{K}_1} e^{-i\beta \hat{S}^{k}_{tt'}}\right) \\
		               & =  e^{-i\beta \hat{S}^1_{tt'}} \left( \bigotimes_{k\in \mathcal{K}_1} e^{-i\beta \hat{S}^{k+1}_{tt'}} e^{-i\beta \hat{P}^{k}_{ttt'}}  e^{-i\beta \hat{S}^{k+1}_{tt'}}\right)  e^{-i\beta \hat{S}^1_{tt'}} \left(\bigotimes_{k\in \mathcal{K}_2} e^{-i\beta \hat{S}^k_{tt'}} e^{-i\beta \hat{P}^{k}_{ttt'}}e^{-i\beta \hat{S}^k_{tt'}}\right)            \\
	\end{aligned}
\end{equation}}
where $\mathcal{K}_1,\mathcal{K}_2$ are defined in \Cref{eq:QE2} and the third equality follows from \Cref{eq:swpexprela}. 

The tensor product over $\mathcal{K}_2$ has support on all {$\hat{S}^k_{tt'}$, $\hat{P}^k_{t'tt'}$ and $\hat{P}^k_{ttt'}$} with $k\in\mathcal{K}_2$. As {$e^{-i\beta \hat{S}^1_{tt'}}$} has support on the identity then the support from the tensor product over $\mathcal{K}_2$ is not lost. In addition {$e^{-i\beta \hat{S}^1_{tt'}}$} yields support on $\hat{S}^1_{tt'}$. Next, the tensor product over $\mathcal{K}_1$ has support on the identity so all prior support is retained. In addition we gain support on all {$\hat{S}^k_{tt'}$, $\hat{P}^k_{t'tt'}$ and $\hat{P}^k_{ttt'}$} with $k\in\mathcal{K}_1$. Then, the factor {$e^{-i\beta \hat{S}^1_{tt'}}$} once again retains all prior support as it has support on the identity. Finally, all support in~\Cref{eq:QE23expan} include the binary arithmetic operations between $t$ and $t'$ asset block as introduced in the~\Cref{sec:mix}. Changing the parameter $\beta$ will influence the coefficients of the unitary on the support and thus control the supposition of the computational basis state.

\section{Simulation on qubit excitation mixing operator with different \texorpdfstring{$L$}{L}}\label{app:simumix}
In this appendix, we compare mixing operator {$\hat{M}_L(\beta)$} with $L=1,2$ provided in~\Cref{sec:mix}.
We observe the loss landscapes corresponding to {$\hat{M}_1(\beta)$} have larger gradients, fewer and deeper valleys with roughly the same minimum value as {$\hat{M}_2(\beta)$}---see data in~\Cref{tab:appsimumix}. A representative example is shown in~\Cref{fig:mixerCom}(left). We also find that the ansatz becomes harder to optimise in the presence of stochastic measurement errors when using {$\hat{M}_2(\beta)$} compared to {$\hat{M}_1(\beta)$}---see~\Cref{fig:mixerCom}(right). {Despite the increase of $\alpha_{\mean}$ with $p$ we find in the body of this article that this does not degrade the output of our algorithm: $\alpha_{\min}$.}
In~\Cref{sec:mix}, we demonstrate that the ansatz wavefunction generated by {$\hat{M}_2(\beta)$} is a superposition with support on more computational basis vectors than generated by {$\hat{M}_1(\beta)$}. The increased support generated by {$\hat{M}_2(\beta)$} suggests that as $\beta$ is varied, the total variation in any given occupation will be less sensitive to $\beta$. Thus, we conjecture, this reduces the gradients with respect to $\beta$ and makes the optimisation harder in the presence of stochastic measurement errors. Therefore, we primarily use {$\hat{M}_1(\beta)$} as our mixing operator. 

  \begin{table}[!ht]
    \centering
    \setlength{\tabcolsep}{6pt}
    \begin{tabular}{lrrrrr}\toprule
      Mixing operator & Average gradient & Largest gradient & Min & Max$-$Min & Valleys  \\ \cmidrule{1-6}
      {$\hat{M}_1(\beta)$} & $0.0298 \pm 0.0025 $ & $0.132\pm 0.011$ & \bm{$0.220 \pm 0.028$} & $0.051 \pm 0.003 $ & $4.2\pm 0.1$  \\ 
      {$\hat{M}_2(\beta)$} & $0.0226 \pm 0.0024 $ & $0.101 \pm 0.010$ & \bm{$0.220 \pm 0.028$} & $0.022 \pm 0.002 $ & $6.3 \pm 0.1$ \\
     \bottomrule
    \end{tabular}
        \caption{Tabulated properties of the loss landscapes of mixing operators {$\hat{M}_L(\beta)$} with $L=1,2$. The data is collected from 100 simulations each using randomly generated DGMVP instances with the asset number $n=4$ and the asset binary block length $l=5$. In our simulations, we use the $|\bm{w}\rangle_{mb}$ as the initial state and a $p=10$ layer ansatz. We scan the landscape of the fifth mixing operator's parameter {$\beta_5\in\left[0,\pi\right]$} with a resolution of $\pi/500$ and fix all the other {cost and mixing rotation parameters as $\pi/4$ and $\pi/8$, respectively}. Each data point is estimated with $N_M$=$2^{22}=4194304$ measurement shots. The average (second column) and largest (third column) gradient magnitudes of the approximation ratio are calculated using finite differences. The minimum (fourth column), maximum$-$minimum (fifth column), and the number of valleys (sixth column) are averaged over the $100$ instances. All values for {$\hat{M}_1(\beta)$} and {$\hat{M}_2(\beta)$} are separated by more than the standard error in the mean except for the minimum (bold).}\label{tab:appsimumix}
  \end{table}

\begin{figure}[h]
	\centering
	\includegraphics{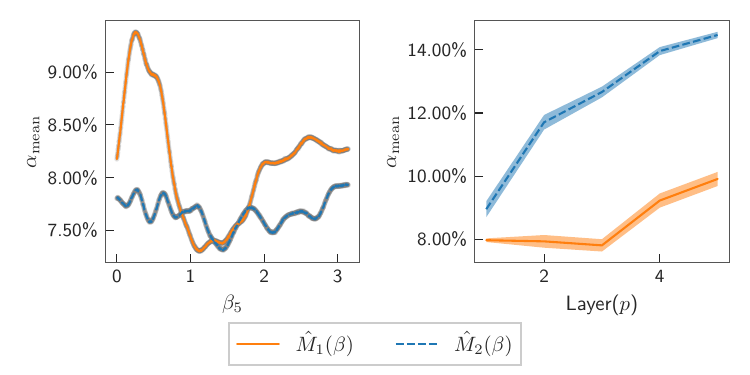}
	\caption[A comparison between designs of the mixing operators.]{               A comparison of the mixing operators {$\hat{M}_1(\beta)$ and $\hat{M}_2(\beta)$}. The left panel shows a typical loss landscape of the parameter, {$\beta_5\in\left[0,\pi\right]$}, of the fifth mixing operator with a resolution of $\pi/500$ with all the other {cost and mixing rotation parameters as $\pi/4$ and $\pi/8$, respectively}. In our simulations, we use the $|\bm{w}\rangle_{mb}$ as the initial state and a $p=10$ layer ansatz. Each plotted point is estimated with $N_M$=$2^{22}=4194304$ measurement shots. The thin grey shadow beneath the coloured lines represents ensembles of estimations, each using $N_M$=$2^{22}=4194304$ measurement shots to indicate the stochastic noise. 
The right panel depicts $\alpha_{\mean}$ for the fix{ed-$p$} ansatz approach with the $|\boldsymbol{w}\rangle_{rd}$ initial state as a function of the number of layers $p$. 
	The initial parameters $\boldsymbol{\beta}, \boldsymbol{\gamma}$ are sampled from a uniform distribution. We utilise COBYLA as the optimser with $N_m = 2^{10}=1024$, and post-optimisation measurements $N_M = 2^{16}=65536$. The shaded regions denote standard error in the mean. For both panels, the asset number $n = 4$ and the asset binary block length $l=5$.
	}\label{fig:mixerCom}
\end{figure}

\section{The performance comparison of qubit excitation mixing operators to parity and XYY mixers }\label{app:mix_compare_parity_xyy}
{In this appendix, we compare our qubit-excitation mixing operator to the a mixer with the same structure but with parity and XYY mixers substituted for the two- and three-qubit excitation operators, respectively---see \Cref{tab:mix_compare}. The results show that the qubit-excitation operators indicate a better performance than the parity and XYY mixer in classical portfolio optimisation. 
 \begin{table}[!ht]
 {
    \centering
    \setlength{\tabcolsep}{6pt}
    \begin{tabular}{lrrr}
\toprule
Hard{-constraint} mixing operator generators & Initial state & $\alpha_{\mean}$ & $\alpha_{\min}$ \\ 
\cmidrule{1-4}
\multirow{2}{*}{Two- and three-qubit excitation} 
& $|\bm{w}\rangle_{mb}$ & $0.071 \pm 0.005$ & $0.0027 \pm 0.0007$ \\
& $|\bm{w}\rangle_{ws}$ & $0.0035 \pm 0.0004$ & $0.0010 \pm 0.0002$ \\ 
\cmidrule{1-4}
\multirow{2}{*}{Parity and XYY mixer~\cite{hadfieldQuantumApproximateOptimization2019,chenQuantumPortfolioOptimization2023}} 
& $|\bm{w}\rangle_{mb}$ & $0.088 \pm 0.006$ & $0.006 \pm 0.001$ \\
& $|\bm{w}\rangle_{ws}$ & $0.0042 \pm 0.0005$ & $0.0010 \pm 0.0001$ \\
\bottomrule
\end{tabular}
        \caption{The performance of hard{-constraint} mixing operators with the same design structures as {$\hat{M}_1(\beta)$} using different generators. The data is collected from 100 simulations each using randomly generated DGMVP instances with the asset number $n=5$ and the asset binary block length $l=5$. In our simulations, we use the $|\bm{w}\rangle_{mb}$ and $|\bm{w}\rangle_{ws}$ as the initial state and a $p=2$ layer ansatz.  The $\alpha_{\mean}$ (third column), $\alpha_{\min}$ (fourth column) are averaged over the $100$ instances. All values are separated by more than one standard error in the mean except the two values for $\alpha_{\min}$ using $|\bm{w}\rangle_{ws}$ as initial state. The classical optimiser is DA with ${\mathcal{J}} = 2000$ and $N_m = 16$, and the optimisation method is unfrozen-layerwise optimisation with random initial ansatz parameters. The final expectation value measurement shots $N_M = 2000$.}\label{tab:mix_compare}
        }
  \end{table}
  }

\section{Comparison of constrained optimisation by linear approximation (COBYLA) and dual annealing (DA)}\label{app:ccoda}
In this appendix, we show a collection of statistical simulations that compare the optimisation performance of COBYLA and DA when optimising multilayered ansätze. In~\Cref{fig:optimiserCom}(a), the number of cost function access $N_f$ of COBYLA increases linearly with $p$, while DA consumes roughly a constant number of cost function accesses because we set a fixed ${\mathcal{J}}$ for DA. Note that ${\mathcal{J}}$ is a soft constraint and can be slightly exceeded as we allow the final optimiser iteration in the local search to complete after the maximal ${\mathcal{J}}$ is violated. We observe that DA uses fewer total cost function accesses than COBYLA to achieve smaller values of $\alpha_{\min}$ and $\alpha_{\mean}$ in the large $p$ regime. We also observe in~\Cref{fig:optimiserCom}(b) and (c) that DA is robust to stochastic measurement noise as it achieves a lower $\alpha_{\mean}$ and a much lower $\alpha_{\min}$ for nearly all examined $p$. {Further, while for both optimisers $\alpha_{\mean}$ increases with $p$, the metric of solution quality, $\alpha_{\min}$, decreases for $p\ge2$.} The experiment shows the advantages of using DA to reduce the cost of quantum resources in optimising parameterised quantum circuits.

\begin{figure}[h]
	\centering
	\includegraphics{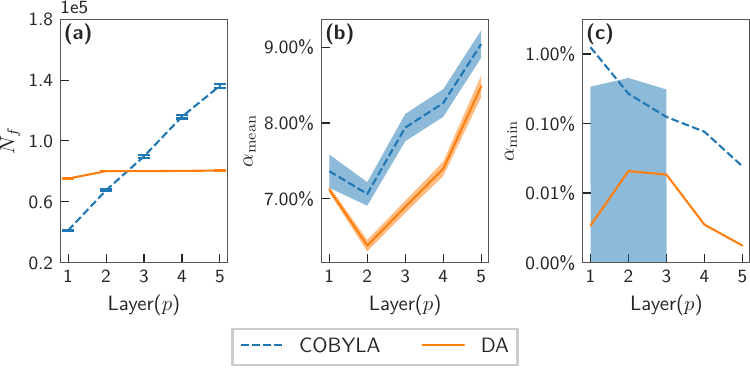}
	\caption[The comparison between dual annealing and COBYLA.]{A performance comparison between DA and COBYLA. The panels from left to right show (a) $N_{f}$, (b) $\alpha_{\mean}$, and (c) $\alpha_{\min}$ as a function of $p$. The results are averaged over 100 simulations with random initial ansatz parameters. We plot $\alpha_{\min}$ on a log scale to emphasise the exponential decay.
	The blue dash line is the COBYLA optimised sample mean with $N_m= 2^{10}=1024$, 
	the orange solid line is the optimised sample mean of DA with ${\mathcal{J}} = 5000$ and $N_m= 2^{4}=16$. The error bar in (a) is induced by the expectation value estimations access, and the shaded region in (b) and (c) indicates the $25\%$ to $75\%$ percentile of data. These simulations use the mixing operator {$\hat{M}_1(\beta)$}  and the initial state $|\boldsymbol{w}\rangle_{mb}$ with $n=l=4$. The optimisation method is the fix{ed-$p$} ansatz approach for every $p$s with the final expectation value measurement shots $N_M=2^{16}=65536$.
 }\label{fig:optimiserCom}
\end{figure}

\section{Further investigations of the hyperparameters impact on performance of DA's optimisation}\label{app:dafuther}
To investigate the performance and identify the best parameters for DA, we compared DA's optimisation performance by varying ${\mathcal{J}}$ and $N_m$. We consider two methods for fixing ${\mathcal{J}}$: constant ${\mathcal{J}}$ and linear growth (LG) in which ${\mathcal{J}}=5000p$. However, we allow the DA to complete its final local search after the maximal ${\mathcal{J}}$ is violated, and so ${\mathcal{J}}$ can be slightly exceeded. 

In~\Cref{fig:da_para_compar}(a), the LG line shows a constant plateau with $p\geq 2$ demonstrating sufficient ${\mathcal{J}}$ and $N_m$ will stabilise the optimisation. The upward trend with constant ${\mathcal{J}}$ is discussed in~\Cref{app:inistat}. When decreasing ${\mathcal{J}}$ and $N_m$, we observe the optimised $\alpha_{\mean}$ will be in general higher, and $\alpha_{\mean}$ increases with increasing $p$.

However, we are interested in the region with small ${\mathcal{J}}$ and $N_m$ where achieving the real quantum advantage may be possible. By fixing ${\mathcal{J}}=2000$ (the dashed lines with round markers in~\Cref{fig:da_para_compar}), we find that increasing $N_m$ from $4$ to $8$ and $8$ to $16$ improves $\alpha_{\mean}$ roughly the same amount as increasing $N_m$ from $16$ to $1024$. Thus, we observe the decrease of $\alpha_{\mean}$ with increasing $N_m$ starts to plateau. Now fixing $N_m =16$ (see the dark-green lines in~\Cref{fig:da_para_compar}), we observe decreasing ${\mathcal{J}}$ from $5000$ to $2000$ has little impact of $\alpha_{\min}$ but a further decrease to $100$ is extremely detrimental. Thus, after $2000$ expectation value estimations, the improvement per expectation value estimation is relatively small. Thus further optimisation is expensive in expectation value estimations.

Another metric of importance is $\alpha_{\min}$, see \Cref{fig:da_para_compar}(b) {which, unlike $\alpha_{\mean}$, characterises the final solution quality of our algorithm}. We observe that $\alpha_{\min}$ often has a local maximum around $p=2$ but generally decreases with increasing $p$.
As was observed in~\Cref{app:inistat}, the parameters of the best $\alpha_{\mean}$ do not correspond to the best $\alpha_{\min}$ directly. Thus, focusing on optimising only one ratio may lead to the detriment of the other. 
We find that the choice of ${\mathcal{J}}=2000$ and $N_m =16$ for DA is a good trade-off between the optimisation costs (${\mathcal{J}}$ and $N_m$) and the optimisation performance as indicated by $\alpha_{\mean}$ and $\alpha_{\min}$ for the values of $n$ and $l$ considered in this article.

\begin{figure}[h]
	\centering
	\includegraphics{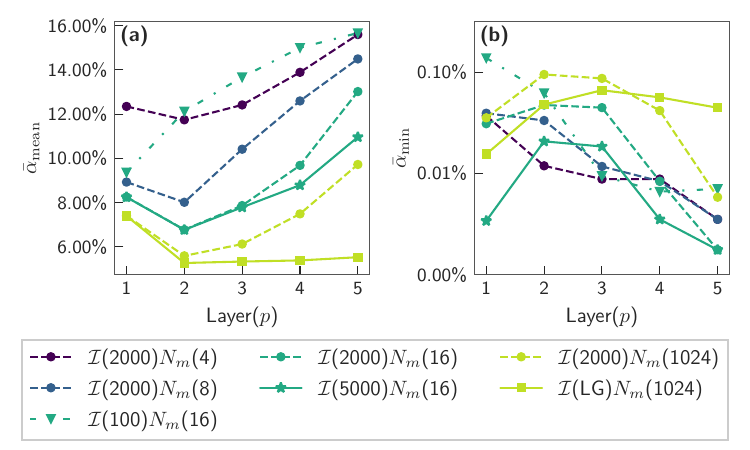}
	\caption[The comparison between dual annealing parameters.]{A comparison of hyperparameters ${\mathcal{J}}$ and $N_m$ impact on the performance of DA's optimisation as a function of the layer $p$: (a) $\alpha_{\mean}$, (b) $\alpha_{\min}$. Every line in the graphs is the average over 100 simulations with random initial parameters.
 These simulations use the initial state $|\boldsymbol{w}\rangle_{mb}$ and the mixing operator {$\hat{M}_1(\beta)$} with $n=l=4$. The optimisation method is the fix{ed-$p$} ansatz approach for every $p$s with the final expectation value measurement shots $N_M=2^{16}=65536$. 
		}\label{fig:da_para_compar}
\end{figure}

\section{Simulation on binary block length and asset size scaling}
\label{app:sim}
To investigate the scaling of our quantum algorithm for solving the DGMVP, we combine the most efficient methods from previous simulations to perform stochastic simulations on problem sets with varied binary block lengths $l$ and the number of assets $n$. For each $l$ and $n$ our simulations contain covariance matrices generated from randomly selected US stocks---see \Cref{app:covariance}.
For the ansatz circuit, we use {$\hat{M}_1(\beta)$} as the mixing operator---see~\Cref{app:simumix}. 
In this appendix, we compare the performance of the initial states $|\bm{w}\rangle_{ws}$ and $|\bm{w}\rangle_{mb}$.   
We apply an unfrozen-layerwise optimisation method (see \Cref{sec:paraopt,sec:layerwise}) up to $p=5$ layers using DA as the classical optimiser. DA uses $N_m = 16$ shots for estimating expectation value during optimisation with a soft constraint ${\mathcal{J}} = 2000$ expectation value estimations for each newly added layer. All results are obtained by reconstructing the statevectors with optimised parameters to remove the stochastic noise from the post-optimisation expectation value induced by estimating the expectation value with shots.

The results are shown in \Cref{fig:app_scaling_1,fig:app_scaling_2,fig:app_scaling_3,fig:app_scaling_4}. We observe that $\alpha_{\mean}^{k\%}$ for the maxbias initial state generally increases with $n$ and $l$---this is expected as the feasible region size increases, but the number of optimisation steps is fixed. On the contrary, for the warm-started initial state $\alpha_{\mean}^{k\%}$ decreases with $n$ and $l$. This is because as $l$ is increased, the DGVMP becomes a closer and closer approximation to the GVMP, and thus the warm-started initial state approaches the global minimum. However, the decrease in $\alpha_{\mean}^{k\%}$ with $n$ is an artefact of post-selecting out covariance matrices for which the warm-started initial state is not the DGVMP solution.

{Finally, we discuss the trends in \Cref{fig:app_scaling_1,fig:app_scaling_2,fig:app_scaling_3,fig:app_scaling_4} with respect to varying $p$. We find $\alpha_{\mean}^{k\%}$ (bottom rows of \Cref{fig:app_scaling_1,fig:app_scaling_2,fig:app_scaling_3,fig:app_scaling_4}) improves with increasing $p$ when the initial state is $|\bm{w}\rangle_{mb}$, but not when the initial state is $|\bm{w}\rangle_{ws}$. However, the quality of our algorithm output is measured by $\mathrm{P}_{gm}$, $\alpha_{\min}$, and $\mathrm{P}_{\min}$---top rows of \Cref{fig:app_scaling_1,fig:app_scaling_2,fig:app_scaling_3,fig:app_scaling_4} left to right. Fortunately, we observe $\mathrm{P}_{gm}$ and $\alpha_{\min}$ consistently improve with increasing $p$. Finally, the trend for $\mathrm{P}_{\min}$ is inconsistent.}

\begin{figure}[h]
	\centering
	\includegraphics{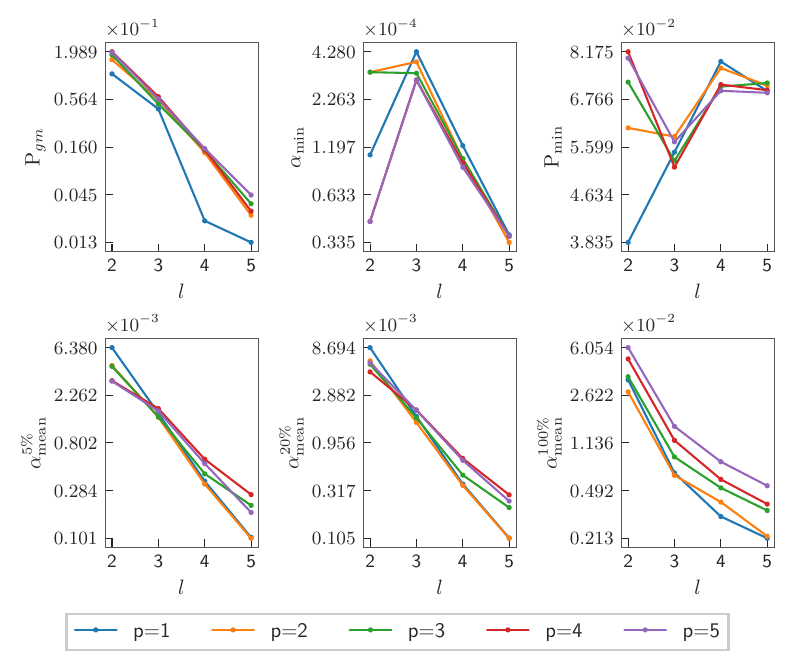}
	\caption[The scaling analysis in l 1]{Stochastic simulations of 100 randomly generated DGMVP instances as a function of binary block lengths $l$ with fixed asset size $n=4$. The initial state is chosen as $|\bm{w}\rangle_{ws}$.
	The top three panels from left to right are the probability of measuring the {optimal} DGMVP solution $\mathrm{P}_{gm}$, the minimal approximation ratio $\alpha_{\mathrm{min}}$, and the probability of measuring the minimal value $\mathrm{P}_{\min}$. The bottom three panels from left to right are the 5, 20, and 100 percentile of the mean approximation ratio, respectively.}
	\label{fig:app_scaling_1}
\end{figure}

\begin{figure}[h]
	\centering
	\includegraphics{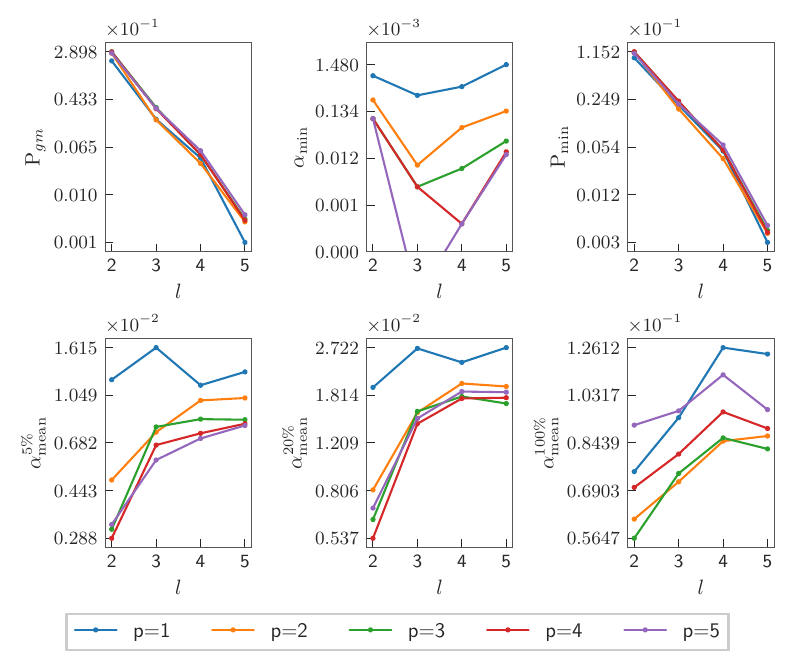}
	\caption[The scaling analysis in l 2]{
    Stochastic simulations of 100 randomly generated DGMVP instances as a function of binary block lengths $l$ with fixed asset size $n=4$. The initial state is chosen as $|\bm{w}\rangle_{mb}$.
	The top three panels from left to right are the probability of measuring the {optimal} DGMVP solution $\mathrm{P}_{gm}$, the minimal approximation ratio $\alpha_{\mathrm{min}}$, and the probability of measuring the minimal value $\mathrm{P}_{\min}$. The bottom three panels from left to right are the 5, 20, and 100 percentile of the mean approximation ratio, respectively. }
	\label{fig:app_scaling_2}
\end{figure}

\begin{figure}[h]
	\centering
	\includegraphics{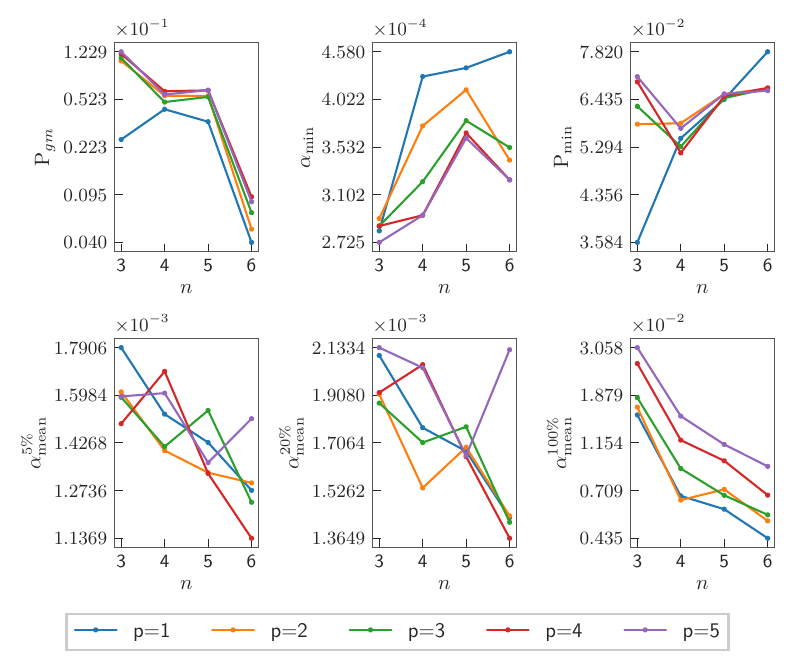}
	\caption[The scaling analysis in n 1]{
 Stochastic simulations of 100 randomly generated DGMVP instances as a function of binary block lengths $n$ with fixed asset size $l=3$. The initial state is chosen as $|\bm{w}\rangle_{ws}$.
	The top three panels from left to right are the probability of measuring the {optimal} DGMVP solution $\mathrm{P}_{gm}$, the minimal approximation ratio $\alpha_{\mathrm{min}}$, and the probability of measuring the minimal value $\mathrm{P}_{\min}$. The bottom three panels from left to right are the 5, 20, and 100 percentile of the mean approximation ratio, respectively.}
	\label{fig:app_scaling_3}
\end{figure}

\begin{figure}[h]
	\centering
	\includegraphics{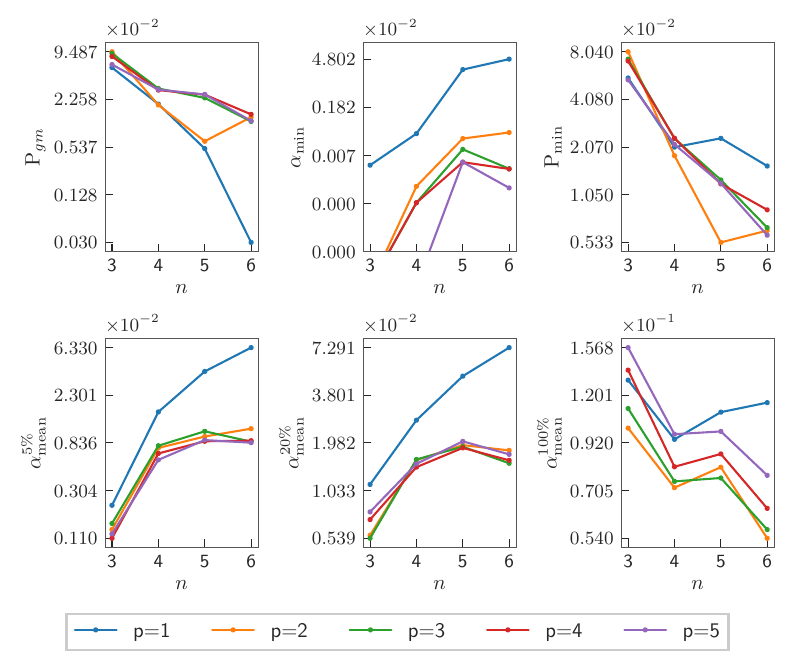}
	\caption[The scaling analysis in n 2]{ Stochastic simulations of 100 randomly generated DGMVP instances as a function of binary block lengths $n$ with fixed asset size $l=3$. The initial state is chosen as $|\bm{w}\rangle_{mb}$.
	The top three panels from left to right are the probability of measuring the {optimal} DGMVP solution $\mathrm{P}_{gm}$, the minimal approximation ratio $\alpha_{\mathrm{min}}$, and the probability of measuring the minimal value $\mathrm{P}_{\min}$. The bottom three panels from left to right are the 5, 20, and 100 percentile of the mean approximation ratio, respectively.}
	\label{fig:app_scaling_4}
\end{figure}

\end{document}